\documentclass[12pt]{article}

\newcommand{\Share}{\mathsf{Share}}
\newcommand{\Eval}{\mathsf{Eval}}
\newcommand{\Rec}{\mathsf{Rec}}
\newcommand{\POLY}{\mathrm{POLY}}
\newcommand{\POLYdmF}{\POLY_{d,m}(\F)}

\usepackage[margin=1.125in]{geometry}
\usepackage{xcolor}
\definecolor{DarkGreen}{rgb}{0.1,0.5,0.1}
\definecolor{DarkRed}{rgb}{0.5,0.1,0.1}
\definecolor{DarkBlue}{rgb}{0.1,0.1,0.5}

\usepackage[small]{caption}
\usepackage[pdftex]{hyperref}
\hypersetup{
    unicode=false,          
    pdftoolbar=true,        
    pdfmenubar=true,        
    pdffitwindow=false,      
    pdfnewwindow=true,      
    colorlinks=true,       
    linkcolor=DarkBlue,          
    citecolor=DarkGreen,        
    filecolor=DarkGreen,      
    urlcolor=DarkBlue,          
    pdftitle={},
    pdfauthor={},    
}
\usepackage{amssymb,amsmath,amsthm,amsfonts,enumitem}
\usepackage{blkarray}
\usepackage{fullpage,nicefrac,comment}
\usepackage{tikz}
\usetikzlibrary{arrows,decorations.pathmorphing,decorations.shapes,patterns,positioning}
\usepackage[ruled,linesnumbered]{algorithm2e}

\newcommand{\cC}{\ensuremath{\mathcal{C}}}

\newcommand{\cL}{\ensuremath{\mathcal{L}}}
\newcommand{\cM}{\ensuremath{\mathcal{M}}}

\newcommand{\F}{{\mathbb F}}

\newcommand{\fm}{{\mathfrak{m}}}

\newcommand{\inabs}[1]{\left|#1\right|}
\newcommand{\abs}[1]{\inabs{#1}}

\newcommand{\inset}[1]{\left\{#1\right\}}

\renewcommand{\epsilon}{\varepsilon}
\newcommand{\vphi}{\varphi}

\newtheorem{theorem}{Theorem} 
\newtheorem{lemma}[theorem]{Lemma} 
\newtheorem{definition}{Definition}

\newtheorem{observation}[theorem]{Observation} 
\newtheorem{conjecture}[theorem]{Conjecture} 
 
\newtheorem{corollary}[theorem]{Corollary} 

\newtheorem{remark}{Remark}
\newtheorem{claim}[theorem]{Claim}

\newtheorem{example}{Example}

\title{A Characterization of Optimal-Rate Linear Homomorphic Secret Sharing Schemes, and Applications}
\author{Keller Blackwell and Mary Wootters}
\date{}

\begin{document}

\maketitle

\begin{abstract}
    
A \emph{Homomorphic Secret Sharing} (HSS) scheme is a secret-sharing scheme that shares a secret $x$ among $s$ servers, and additionally allows an output client to reconstruct some function $f(x)$, using information that can be locally computed by each server.  A key parameter in HSS schemes is \emph{download rate}, which quantifies how much information the output client needs to download from each server.
Recent work \emph{(Fosli, Ishai, Kolobov, and Wootters, ITCS 2022)} established a fundamental limitation on the download rate of linear HSS schemes for computing low-degree polynomials, and gave an example of HSS schemes that meet this limit.

In this paper, we further explore optimal-rate linear HSS schemes for polynomials.  Our main result is a complete characterization of such schemes, in terms of a coding-theoretic notion that we introduce, termed \emph{optimal labelweight codes}.  We use this characterization to answer open questions about the \emph{amortization} required by HSS schemes that achieve optimal download rate.  In more detail, the construction of Fosli et al. required amortization over $\ell$ instances of the problem, and only worked for particular values of $\ell$. We show that---perhaps surprisingly---the set of $\ell$'s for which their construction works is in fact nearly optimal, possibly leaving out only \emph{one} additional value of $\ell$.  We show this by using our coding-theoretic characterization to prove a necessary condition on the $\ell$'s admitting optimal-rate linear HSS schemes. We then provide a slightly improved construction of optimal-rate linear HSS schemes, where the set of allowable $\ell$'s is optimal in even more parameter settings.  Moreover, based on a connection to the MDS conjecture, we conjecture that our construction is optimal for all parameter regimes.
\end{abstract}

\section{Introduction}\label{sec:intro}
A \emph{Homomorphic Secret Sharing} (HSS) scheme is a secret sharing scheme that supports computation on top of the shares~\cite{C:Benaloh86a, BGI16, BGILT18}.
Homomorphic Secret Sharing has found many applications in recent years, from private information retrieval to secure multiparty computation (see, e.g., \cite{BCGIO17, BGILT18}).

In more detail, a standard $t$-private (threshold) secret sharing scheme shares an input $x$ as $s$ \emph{shares}, $\Share(x) = (y_1, \ldots, y_s)$, which are then distributed among $s$ servers; the goal is that any $t+1$ of the servers can together recover the secret $x$, while no $t$ of the servers can learn anything about $x$.\footnote{In this work, we focus on \emph{information-theoretic} security, so the above statement means that the joint distribution of any $t$ shares does not depend on the secret $x$.}

A $t$-private \emph{HSS} scheme has the additional feature that the servers are able to compute functions $f$ in some function class $\mathcal{F}$, as follows.  Each server $j$ does some local computation on its share $y_j$ to obtain an \emph{output share} $z_j = \Eval(f,j,y_j)$.  The output shares $z_1, \ldots, z_s$ are then sent to an \emph{output client}, who uses these output shares to recover $f(x) = \Rec(z_1, \ldots, z_s)$.  Formally, we say that the HSS scheme $\pi$ is given by the tuple of functions $(\Share, \Eval, \Rec)$ (see Definition~\ref{def:HSS}).  

In order for this notion to be interesting, the output shares $z_j$ should be substantially smaller than the original shares $y_j$; otherwise any $t+1$ servers could just communicate their entire shares $y_j$ to the output client, who would recover $x$ and then compute $f(x)$.  To that end, prior work~\cite{FIKW22} focused on the the \emph{download rate} of (information-theoretically secure) HSS schemes.  The download rate (see Definition~\ref{def:downloadrate}) of an HSS scheme is the ratio of the number of bits in the output $f(x)$ (that is, the number of bits the output client wants to compute), to the number of bits in all of the output shares $z_j$ (that is, the number of bits that the output client downloads); ideally this rate would be as close to $1$ as possible.

The work~\cite{FIKW22} focused on HSS schemes for the function class $\mathcal{F} = \POLY_{d,m}(\F)$, the class of all $m$-variate, degree-$d$ polynomials over a finite field $\F$, and we will do the same here.  
One of the main results of that work was an infeasibility result on the download rate for \emph{linear} HSS schemes, which are schemes where both $\Share$ and $\Rec$ are linear over some field; note that $\Eval$ (which converts the shares $y_j$ to the output shares $z_j$) need not be linear.

\begin{theorem}[\cite{FIKW22} (informal, see Theorem~\ref{thm:FIKW rate LB})]\label{thm:FIKWinformal}
    Any $t$-private $s$-server linear HSS scheme for $\text{POLY}_{d,m}(\mathbb{F})$ has download rate at most $(s - dt)/s$.
\end{theorem}
In fact, \cite{FIKW22} proved that the bound $(s-dt)/s$ holds even if the HSS scheme is allowed to ``amortize'' over $\ell$ instances of the problem.  That is, given $\ell$ secrets $x^{(1)}, \ldots, x^{(\ell)}$, shared \emph{independently} as $y_j^{(i)}$ for $j \in [s], i \in [\ell]$, the $j$'th server may do local computation on its shares $y_j^{(1)}, \ldots, y_j^{(\ell)}$ to compute an output share $z_j$; the output client needs to recover $f_1(x^{(1)}), \ldots, f_\ell(x^{(\ell)})$, given $z_1, \ldots, z_s$.  

Moreover, \cite{FIKW22} complemented this infeasibility result with a construction achieving download rate $(s - dt)/s$, provided that the amortization factor $\ell$ is a sufficiently large multiple of $(s-dt)$:
\begin{theorem}[\cite{FIKW22}] \label{thm:FIKWInformalUB}
Suppose that $s > dt$, and suppose that $\ell = j(s-dt)$ for some integer $j \geq \log_{|\F|}(s)$.
Then there is a $t$-private, $s$-server linear HSS scheme for $\POLYdmF$ (with CNF sharing, see Definition~\ref{def:CNF}), download rate at least $(s-dt)/s$, and amortization parameter $\ell$.
\end{theorem}

This state of affairs leaves open two questions, which motivate our work:
\begin{itemize}
    \item[(1)] Are there optimal-rate linear HSS schemes for $\POLYdmF$ beyond the example in Theorem~\ref{thm:FIKWInformalUB}?  Can we characterize them?
    \item[(2)] Is some amount of amortization necessary to achieve the optimal rate?  If so, which amortization parameters $\ell$ admit optimal-rate linear HSS schemes for $\POLYdmF$?
\end{itemize}

\subsection{Main Results}\label{sec:results}

We answer both questions (1) and (2) above. 
For all of our results, we consider \emph{CNF sharing}~\cite{ISN89} (see Definition~\ref{def:CNF}).  It is known that CNF sharing is universal for linear secret sharing schemes, in that $t$-CNF shares can be locally converted to shares of \emph{any} linear $t$-private secret sharing scheme~\cite{CDI05}.

\paragraph{High-level answers to Questions (1) and (2).}
\begin{itemize}
\item[(1)] Our results give a {complete characterization} of the $\Rec$ functions for optimal-rate linear HSS schemes for $\POLYdmF$.  Our characterization is in terms of linear codes with optimal \emph{labelweight}, a notion that we introduce, which generalizes distance.  Similar to how optimal-distance (MDS) codes are characterized by a property of their generator matrices, we are also able to characterize the generator matrices of optimal-labelweight codes.

We hope that the notion of labelweight will find other uses, and we believe that this characterization will be useful for studying linear HSS schemes beyond our work.

\item[(2)] Using our characterization, we show that the amortization parameters $\ell$ given in Theorem~\ref{thm:FIKWInformalUB} are in fact (usually) \emph{all} of the admissible $\ell$'s.\footnote{The parenthetical ``(usually)'' is due to a corner case; it is possible that $\ell = j(s - dt)$ with $j = \lceil \log_{|\F|}(s) \rceil - 1$, which is one smaller than the bound given in Theorem~\ref{thm:FIKWInformalUB}.}  In particular, some amortization is required, and there are parameter regimes where one cannot obtain optimal download rate with amortization parameters $\ell$ other than those in the construction in \cite{FIKW22}. Moreover, we give a construction that \emph{slightly} improves on the construction in \cite{FIKW22}, closing the gap between the lower and upper bounds in a few more parameter regimes.  

We find this answer to Question (2) somewhat surprising; we expected that the particular form of $\ell$ in Theorem~\ref{thm:FIKWInformalUB} was an artifact of the construction, but it turns out to be intrinsic.
\end{itemize}

 We describe our results in more detail below.

 \paragraph{(1) A characterization of optimal-rate linear HSS schemes.}  Our main result is a characterization of linear HSS schemes for $\POLYdmF$ with optimal download rate $(s - dt)/s$.  In particular, we show that the $\Rec$ algorithms for such schemes (with CNF sharing) are \emph{equivalent} to a coding-theoretic notion that we introduce, which we refer to as codes with \emph{optimal labelweight.}   
 \begin{definition}[Labelweight]\label{def:LW}
     Let $\cC \subseteq \F^n$ be a linear code of dimension $\ell$.  Let $\mathcal{L}:[n] \to [s]$ be any function, which we refer to as a \emph{labeling} function.  The \emph{labelweight} of $\mathbf{c} \in \cC$ is the number of distinct labels that the support of $\mathbf{c}$ touches:
     \[ \Delta_{\cL}(\mathbf{c}) = |\{ \cL(i) \,: \, i \in [n], c_i \neq 0 \}|.\]
     The \emph{labelweight} of $\cC$ is the minimum labelweight of any nonzero codeword:
     \[ \Delta_{\cL}(\cC) = \min_{c \in \cC\setminus \{0\}} \Delta_{\cL}(\mathbf{c}).\]
 \end{definition}
 In particular, if $s = n$ and $\mathcal{L}(j) = j$ for all $j \in [n]$, then $\Delta_\cL(\cC)$ is just the minimum Hamming distance of $\cC$.  Thus, the labelweight of a code generalizes the standard notion of distance.

Our main characterization theorem is the following.
\begin{theorem}[Optimal linear HSS schemes are equivalent to optimal labelweight codes. (Informal, see Lemma~\ref{lem:HSS_implies_lblwt} and Theorem~\ref{thm:lblwtImpliesHSS})]\label{thm:maininformal}
Let $\pi = (\Share, \Eval, \Rec)$ be a $t$-private, $s$-server linear HSS for $\POLYdmF$, with download rate $(s - dt)/s$.  Let $G$ be the matrix that represents $\Rec$ (see Observation~\ref{obs: linear reconstruction is matrix multiplication}).  Then there is some labeling function $\cL$ so that $G$ is the generator matrix for a code $\cC$ with rate $(s-dt)/s$ and with $\Delta_{\cL}(\cC) \geq dt + 1$.

Conversely, suppose that there is a labeling function $\cL:[n] \to [s]$ and a linear code $\cC \subseteq \F^n$ with rate $(s-dt)/s$ and with $\Delta_{\cL}(\cC) \geq dt + 1$.  Then any generator matrix $G$ of $\cC$ describes a linear reconstruction algorithm $\Rec$ for an $s$-server $t$-private linear HSS for $\POLYdmF$ that has download rate $(s - dt)/s$.
\end{theorem}

We remark that the converse direction is constructive: given the description of such a code $\cC$, the proof (see Theorem~\ref{thm:lblwtImpliesHSS}) gives an efficient construction of the $\Eval$ function as well as the $\Rec$ function.

Given Theorem~\ref{thm:maininformal}, we now want to know when optimal labelweight codes exist.  Our second main result---which may be of independent interest---characterizes these codes in terms of their generator matrices.  We describe this characterization more in the Technical Overview (Section~\ref{sec:tech}), and next we describe the implications for HSS schemes.

\paragraph{(2) Understanding the amortization parameter $\ell$.}

Using our characterization above, we are able to nearly completely classify which amortization parameters $\ell$ admit download-optimal linear HSS schemes for $\POLYdmF$.
Given Theorem~\ref{thm:maininformal}, it suffices to understand when optimal labelweight codes exist.  We show the following theorem.

\begin{theorem}[Limitations on optimal labelweight codes.  (Informal, see Theorem~\ref{thm:lblwt_limitations})]\label{thm:ell_LB_informal}
    Suppose that $G \in \F_q^{\ell \times n}$ is the generator matrix for a code $\cC$ with rate $\ell/n = (s-dt)/s$, and suppose that there is a labeling function $\cL:[n] \to [s]$ so that $\Delta_{\cL}(\cC) \geq dt + 1$.  Then $\ell = j(s -dt)$ for some integer $j$, with $j \geq \max\{ \log_{q}(s - dt + 1), \log_{q}(dt + 1) \}$.
\end{theorem}

This result
 \emph{nearly} matches the feasibility result of \cite{FIKW22} (Theorem~\ref{thm:FIKWInformalUB} above).  That is, Theorem~\ref{thm:FIKWInformalUB}
allows for $\ell = j(s-dt)$ for any integer $j \geq \log_{q}(s)$.
As we always have $\max(s - dt + 1, dt + 1) \geq s/2$, the conclusion in Corollary~\ref{thm:ell_LB_informal} implies that $j \geq \log_{q}(s/2) = \log_{q}(s) - 1/\log_2(q)$.  As $j$ must be an integer, this exactly matches the amortization parameter in Theorem~\ref{thm:FIKWInformalUB} whenever 
\[ \left\lceil \frac{ \log(s) - 1 }{\log(q) } \right\rceil = \left\lceil \frac{ \log (s) }{ \log(q) }  \right \rceil, \]
which holds for most values of $s$ and $q$ when $q$ is large. Moreover, we give a construction that \emph{very} slightly improves on the one from Theorem~\ref{thm:FIKWInformalUB}, which exactly matches Theorem~\ref{thm:ell_LB_informal} for even more settings of $s$:

\begin{theorem}[Construction of optimal labelweight codes. (Informal, see Theorem~\ref{thm:lblwt})]\label{thm:lblwt_Informal}
Let $j$ be any integer so that
\[ j \geq \begin{cases}
            \log_q(s-1) & q^j \text{ odd, or } (s-dt) \not\in \lbrace 3, q^j-1 \rbrace\\
            \log_q(s-2) & q^j \text{ even, and } (s-dt) \in \lbrace 3, q^j-1 \rbrace
        \end{cases}.
        \]
There is an explicit construction of a code $\cC \subseteq \F_q^n$ with dimension $\ell = j(s-dt)$, block length $n = js$, (and hence rate $(s-dt)/s$), and a labeling function $\cL:[n] \to [s]$ so that $\Delta_\cL(\cC) \geq dt + 1$.
\end{theorem}
Given Theorem~\ref{thm:maininformal} (our equivalence between codes with good labelweight and linear HSS schemes), Theorem~\ref{thm:lblwt_Informal} immediately implies that there is a download-optimal linear HSS scheme for $\POLYdmF$ with amortization parameter $\ell = j(s-dt)$ for any $j$ as in Theorem~\ref{thm:lblwt_Informal}.  (See Corollary~\ref{cor:amortization} for a formal statement).

The HSS result implied by Theorem~\ref{thm:lblwt_Informal} is quite close to the existing result from \cite{FIKW22} (Theorem~\ref{thm:FIKWInformalUB}).\footnote{In addition to the quantitative results being quite close, the constructions themselves are also similar.  This is not surprising, given that one of our main results is that such schemes must be extremely structured.  We discuss the relationship between our construction and that of \cite{FIKW22} in the technical overview below.}  The only quantitative difference between Theorem~\ref{thm:lblwt_Informal} and the result of \cite{FIKW22} is that we can take $j$ to be either $\lceil \log_q(s-1) \rceil$ or $\lceil \log_q(s-2) \rceil$, rather than $\lceil \log_q(s) \rceil$.  Meanwhile, Theorem~\ref{thm:ell_LB_informal} says that we must have $j \geq \lceil \log_q(s/2 + 1) \rceil$.  So there are a few values of $s$ where Theorem~\ref{thm:lblwt_Informal} is tight (matching Theorem~\ref{thm:ell_LB_informal}) but Theorem~\ref{thm:FIKWInformalUB} (the construction from~\cite{FIKW22}) is not.  For example, if $q = 2$ and $s = 2^r + 1$, then 
\[ \lceil \log_q(s-2) \rceil = \lceil \log_q(s/2 + 1) \rceil = r,\]
but $\lceil \log_q(s) \rceil = r + 1$ is larger.

There are still a few parameter regimes where there is a gap (of size one) between the $j$'s that work in Theorem~\ref{thm:lblwt_Informal} and bound of Theorem~\ref{thm:ell_LB_informal}.  We conjecture that in fact the \emph{construction} is optimal, and Theorem~\ref{thm:ell_LB_informal} is loose.
In more detail, in Section~\ref{sec:MDS}, we establish a connection between constructions of download-optimal HSS schemes and the \emph{MDS conjecture} over extension fields, which  may be of independent interest.  Further, we show that, assuming the MDS conjecture, our construction of Theorem~\ref{thm:lblwt_Informal} is optimal within a natural class of constructions, even when it does not match our infeasibility result in Theorem~\ref{thm:ell_LB_informal}.

\subsection{Technical Overview}\label{sec:tech}

In this section, we give a high-level overview of our techniques.

\paragraph{Relationship between download-optimal HSS schemes and optimal labelweight codes.}

Theorem~\ref{thm:maininformal} states that download-optimal linear HSS schemes are equivalent to optimal-labelweight codes.  To give some intuition for the connection, we explain the forward direction, which is simpler.  For simplicity, suppose that the function $f$ is identity, amortized over $\ell$ secrets $x^{(1)}, \ldots, x^{(\ell)} \in \F$. 
 This problem is called \emph{HSS for concatenation} in \cite{FIKW22}; the goal is to share the $\ell$ secrets independently, and then communicate them to the output client using significantly less information than simply transmitting $t+1$ shares of each secret.
 
 Consider a linear HSS scheme for this problem.  By definition, the reconstruction algorithm $\Rec$ is linear, and so can be represented by a matrix $G \in \F^{\ell \times n}$, where $n$ is the total number of symbols of $\F$ sent by all of the servers together.  In more detail, we view each output share $z_j$ as a vector over $\F$, and concatenate all of them to obtain a vector $\mathbf{z} \in \F^n$.  Then we can linearly recover the $\ell$ concatenated secrets from $\mathbf{z}$:
\[ \begin{pmatrix} x^{(1)} \\ \vdots \\ x^{(\ell)} \end{pmatrix} = \Rec(z_1, \ldots, z_s) = G\mathbf{z}. \]  

Now we can define a labeling function $\mathcal{L}:[n] \to [s]$ so that $\mathcal{L}(r) \in [s]$ is the identity of the server sending the $r$'th symbol in $\mathbf{z}$.  We claim that $G$ is the generator matrix\footnote{When we say that $G$ is the \emph{generator matrix} of $\cC \subseteq \F^n$, we mean that $\cC$ is the rowspan of $G$.} of a code $\cC$ with $\Delta_{\cL}(\cC) \geq t + 1$.  (Recall that in HSS for concatenation, $d = 1$, so $t+1 = dt+1$ is the bound on $\Delta_{\cL}(\cC)$ that we want in Theorem~\ref{thm:maininformal}).  To see this, suppose that there were some codeword $\mathbf{m}^T G$ for some $\mathbf{m} \in \F^\ell$ with labelweight at most $t$.  But consider the quantity
\[ X := \sum_{i=1}^\ell m_i x^{(i)} = \mathbf{m}^T G \mathbf{z}. \]
If $\mathbf{m}^T G$ had labelweight at most $t$ (with this labeling $\mathcal{L}$), that means that the quantity $X$ can be computed using messages coming from only $t$ of the servers.
But this contradicts the $t$-privacy of the original secret sharing scheme.  Indeed, suppose without loss of generality that $m_1 \neq 0$.  Then if $x^{(2)} = \cdots = x^{(\ell)} = 0$, some set of $t$ servers would be able to recover $X = m_1 x^{(1)}$ and hence $x^{(1)}$, and this should be impossible.
For the full proof of this direction, we need to generalize to larger $d$'s and more general functions, but the basic idea is the same.  

The converse, showing that any optimal-labelweight code $\cC$ implies an optimal HSS scheme, is a bit trickier.  The main challenge is that the connection above tells us what $\Rec$ should be---it should be given by the generator matrix of $\cC$---but it does not tell us what $\Eval$ should be, or that an appropriate $\Eval$ function even exists.  To find $\Eval$, we view the output shares as vectors of polynomials in the input shares.  That is, each server must return some function of the shares that it holds, and over a finite field, every function is a polynomial.  In this view, we can set up an affine system to solve for $\Eval$, where the variables are the coefficients that appear in each server's output polynomials.  Then we show that this system has a solution, and that this solution indeed leads to a legitimate $\Eval$ function.

\paragraph{Characterization of Optimal Labelweight Codes.}
In order to understand when optimal labelweight codes exist, we give a characterization of them in terms of their generator matrices. We show that the generator matrices of optimal labelweight codes can be taken to be \emph{block-totally-nonsingular}.  We defer the formal definition to Definition~\ref{def:blockMDS}, but informally, a block-totally-nonsingular matrix is made up of $j \times j$ invertible blocks $A \in GL(\F, j)$, with the property that any square sub-array of blocks (not necessarily contiguous) is full-rank.
\begin{lemma}[Characterization of optimal label-weight codes. (Informal, see Lemma~\ref{lem:BTN_iff_goodLabelWt})]\label{lem:charblockmat_informal}

Suppose that $\cC \subseteq \F^n$ is a code of rate $(s - dt)/s$, and suppose that there is a labeling function $\cL:[n] \to [s]$ so that $\Delta_{\cL}(\cC) \geq dt + 1$.  Then, up to a permutation of the coordinates, there is a generator matrix $G$ for $\cC$ that looks like 
\[ G = [ I | A ] \]
where $A$ is block-totally-nonsingular.  Conversely, any such matrix is the generator matrix for a code $\cC$ with $\Delta_{\cL}(\cC) \geq dt + 1$, where $\cL:[n] \to [s]$ is the labeling function $\cL(x) = \lceil x/j \rceil$.
\end{lemma}

To give some intuition for the lemma, we first explain where the ``block'' structure comes from. We show (Lemma~\ref{lem: dl is balanced}) that in fact any labeling function $\cL$ so that $\Delta_{\cL}(\cC) = dt + 1$ must be \emph{balanced}, meaning that each set $\cL^{-1}(i)$ for $i \in [s]$ has the same cardinality, $|\cL^{-1}(i)| = j$, for some integer $j$.  The blocks then correspond to the sets $\mathcal{L}^{-1}(i)$ for $i \in [s]$. 

Next, we give some intuition for the ``totally non-singular'' part.  The basic idea is that if there were a square sub-array of blocks that were singular, then there would be a vector $\mathbf{m} \in \F^\ell$ with support on only the relevant blocks, so that $\mathbf{m}^T A = 0$.  If $G = [ I | A ]$ as in Lemma~\ref{lem:charblockmat_informal}, then this implies that  $\mathbf{m}^T G$ has support only on the labels that appear on the first $\ell$ columns of $G$, which it turns out is small enough to contradict the optimal labelweight property.  We refer to the proof of Lemma~\ref{lem:BTN_iff_goodLabelWt} for details on both points.

\paragraph{Applications to Amortization.}
We apply the machinery described above to Question (2), about what amortization parameters $\ell$ are possible, as follows.  As described above, we show that optimal-download-rate HSS schemes are equivalent to optimal labelweight codes, which in turn are equivalent to to block-totally-nonsingular matrices.  At this point, we already know that the amortization parameter $\ell$ must be equal to $j  (s - dt)$ for some integer $j$; this follows from the $j \times j$ block structure of block-totally-nonsingular matrices. It remains only to understand for which $j$'s such matrices exist.  A limitation on such matrices follows from a standard counting argument, and this implies our infeasibility result, given formally as Theorem~\ref{thm:lblwt_limitations}.

As noted above, there is already a near-optimal construction of HSS schemes in \cite{FIKW22} (Theorem~\ref{thm:FIKWInformalUB}).  In our language, that construction is based on systematic generator matrices of \emph{Reed-Solomon Codes.}  After an appropriate conversion to block form, such matrices are of the form $[I|A]$, where $A$ is a block-totally-nonsingular matrix.  Up to the conversion to block-form, the famous \emph{MDS Conjecture} (Conjecture~\ref{conj: MDS conjecture}) implies that Reed-Solomon codes are \emph{nearly} optimal in this context, but it is known that a slight improvement is possible: instead of requiring $j \geq \log_q(s)$ as in Theorem~\ref{thm:FIKWInformalUB}, one can get $j \geq \log_q(s-2)$  or $\log_q(s-1)$, as in Theorem~\ref{thm:lblwt_Informal}.
The change is simple: one essentially adds one or two more (block) columns to the generator matrix (see Appendix~\ref{sec: rs with extra col (apx)} and the proof of Theorem~\ref{thm:lblwt}).
Thus, we can slightly improve on the construction of \cite{FIKW22} (Theorem~\ref{thm:FIKWInformalUB}) by making this slight improvement to the underlying code.

\subsection{Related Work}\label{sec:related}
Linear HSS schemes (for, e.g. low-degree polynomials) are implicit in classical protocols for tasks like secure multiparty computation and private information retrieval~\cite{C:Benaloh86a,STOC:BenGolWig88,STOC:ChaCreDam88,EC:CraDamMau00,BeaverF90,BeaverFKR90,Chor:1998:PIR:293347.293350}. More recently, \cite{BGILT18} initiated the systematic study of HSS, and in there has been a long line of work on the topic, most of which has focused on HSS schemes that are \emph{cryptographically} secure~\cite{BoyleGI15,DHRW16,BGI16a,BGI16,FazioGJS17,BGILT18,BoyleKS19,BCGIKS19,CouteauM21,OrlandiSY21,RoyS21,DIJL23}.  In contrast, we focus on \emph{information-theoretic security}.  The information-theoretic setting was explored in~\cite{BGILT18} and was further studied in \cite{FIKW22}, which is the closest to our work and also our main motivation.  In particular, \cite{FIKW22} focused the download rate of information-theoretically secure HSS schemes (both linear and non-linear), but did not focus quantitatively on the amortization parameter $\ell$.  In contrast, we restrict our attention to linear schemes, but focus on characterizing such schemes and on pinning down $\ell$. 

We note that the work \cite{FIKW22} also obtained linear HSS schemes using coding-theoretic techniques, but the connection that they exploited is different.  In particular, they show that for the case of $d=1$ (that is, \emph{HSS for concatenation}), the existence of linear HSS schemes for $\POLY_{d=1,m}(\F)^\ell$ is equivalent to the existence of linear codes with a particular rate and distance.  However, this characterization only works when $d=1$.  In contrast, our characterization, in terms of codes with good \emph{labelweight}, a generalization of distance, applies for general $d$.

Finally, we mention the \emph{MDS conjecture}, which is related to our results.  The MDS conjecture (Conjecture~\ref{conj: MDS conjecture}) was stated by Segre in 1955~\cite{Segre1955CurveRN}, and roughly says that Reed-Solomon codes have the best alphabet size possible for any \emph{Maximum-Distance Separable} (MDS) code.  After being open for over 60 years, the conjecture has been proved for \emph{prime} order fields, and particular extension fields~\cite{Ball2012OnSO}.  Our characterization of optimal-rate HSS schemes leads us to consider the related question of the best alphabet size possible for any \emph{Totally Nonsingular} matrix, which we observe in Section~\ref{sec:MDS} is equivalent to the MDS conjecture over extension fields.  This connection leads us to conjecture that our construction is in fact optimal, and we show that the MDS conjecture implies that it is, within a natural class of constructions.  See the discussion in Section~\ref{sec:MDS} for more on this connection. 

\subsection{Open Questions and Future Directions}\label{sec:next}

Before we get into the details, we take a moment to highlight some open questions. 
\begin{itemize}
    \item The most obvious question that our work leaves open is about the edge cases in the characterization of amortization parameters $\ell$ that are admissible for optimal-download linear HSS schemes.  Given the relationship to the MDS conjecture over extension fields, resolving this may be quite a hard problem.  However, there is some hope.  In particular, the MDS conjecture would imply that our construction is optimal only within a natural class of constructions; it may be possible to get improved results by leaving that class. 
    
    \item Another open question is to find further applications of our characterization of download optimal HSS schemes.  We use this characterization to nearly resolve the question of what amortization parameters are admissible for such schemes, but we hope that it will be useful for other questions about linear information-theoretic HSS.
    \item Finally, it would be interesting to extend our results to linear HSS schemes that do not have optimal download rate.  For example, perhaps it is possible to get a drastic improvement in the amortization parameter $\ell$ by backing off from the optimal download rate by only a small amount.
\end{itemize}

\subsection{Organization}
In Section~\ref{sec:prelim}, we set notation and record a few formal definitions that we will need.  In Section~\ref{sec:equiv}, we show that download-optimal HSS schemes are equivalent to codes with good labelweight: Lemma~\ref{lem:HSS_implies_lblwt} establishes that HSS schemes imply codes with good labelweight, and Theorem~\ref{thm:lblwtImpliesHSS} establishes the converse.  In Section~\ref{sec:codes}, we characterize the generator matrices of codes with good labelweight (Lemma~\ref{lem:BTN_iff_goodLabelWt}), and use this to prove bounds on good labelweight codes (Theorems~\ref{thm:lblwt} and \ref{thm:lblwt_limitations}); this implies Corollary~\ref{cor:amortization}, which gives nearly-tight bounds on the amortization parameter $\ell$ in download-optimal linear HSS schemes for $\POLYdmF$.

\section{Preliminaries}\label{sec:prelim}

We begin by setting notation and the basic definitions that we will need throughout the paper.

\textbf{Notation.} For $n \in \mathbb{Z}^+$, we denote by $[n]$ the set $\lbrace 1, 2, \ldots, n \rbrace$. We use bold symbols (e.g., $\mathbf{x}$) to denote vectors. For an object $w$ in some domain $\mathcal{W}$, we use $\|w\| = \log_2 (\abs{\mathcal{W}})$ to denote the number of bits used to represent $w$.

\subsection{Homomorphic Secret Sharing}

We consider homomorphic secret sharing (HSS) schemes with $m$ inputs and $s$ servers; each input is shared independently. We denote by $\mathcal{F} = \lbrace f: \mathcal{X}^m \to \mathcal{O} \rbrace$ the class of functions we wish to compute, where $\mathcal{X}$ and $\mathcal{O}$ are input and output domains, respectively. 

\begin{definition}[HSS]\label{def:HSS}

Given a collection of $s$ servers and a function class $\mathcal{F}= \lbrace f: \mathcal{X}^m \to \mathcal{O} \rbrace$, consider a tuple $\pi = (\Share, \Eval, \Rec)$, where $\Share: \mathcal{X} \times \mathcal{R} \to \mathcal{Y}^s$, $\Eval: \mathcal{F} \times [s] \times \mathcal{Y} \to \mathcal{Z}^\ast$, and $\Rec: \mathcal{Z}^\ast \to \mathcal{O}$ as follows\footnote{By $\mathcal{Z}^\ast$, we mean a vector of some number of symbols from $\mathcal{Z}$.}:

\begin{itemize}
    \item[$\bullet$] $\Share(x_i,r_i)$: For $i \in [m]$, $\Share$ takes as input a secret $x_i \in \mathcal{X}$ and randomness $r_i \in \mathcal{R}$; it outputs $s$ 
 shares $\left( y_{i,j} : j \in [s] \right) \in \mathcal{Y}^s$. We refer to the $y_{i,j}$ as \emph{input shares}; server $j$ holds shares $(y_{i,j} 
 : i \in [m])$.

    \item[$\bullet$] $\Eval\left(f, j, \left( y_{1,j}, y_{2,j}, \ldots, y_{m,j} \right) \right)$: Given $f \in \mathcal{F}$, server index $j \in [s]$, and server $j$'s input shares $\left( y_{1,j}, y_{2,j}, \ldots, y_{m,j} \right)$, $\Eval$ outputs $z_j \in \mathcal{Z}^{n_j}$, for some $n_j \in \mathbb{Z}$. We refer to the $z_j$ as \emph{output shares}.

    \item[$\bullet$] $\Rec(z_1, \ldots, z_s)$: Given output shares $z_1, \ldots, z_s$, $\Rec$ computes $f(x_1, \ldots, x_m) \in \mathcal{O}$. 

\end{itemize}

\noindent
We say that $\pi = (\Share, \Eval, \Rec)$ is a $s$-server HSS scheme for $\mathcal{F}$ if the following requirements hold:

\begin{itemize}
    \item[$\bullet$] \textbf{Correctness:} For any $m$ inputs $x_1, \ldots, x_m \in \mathcal{X}$ and $f \in \mathcal{F}$, 

    \begin{equation*}
        \Pr_{\mathbf{r} \in \mathcal{R}^m}\left[ \Rec(z_1, \ldots, z_s) = f(x_1, \ldots, x_m) : 
            \begin{aligned} 
                &\forall i \in [m], \; \left( y_{i,1}, \ldots, y_{i,s} \right) \leftarrow \Share(x_i,r_i)\\
                &\forall j \in [s], \; z_j \leftarrow \Eval\left( f, j, (y_{1,j}, \ldots, y_{m,j}) \right)
            \end{aligned}
         \right] = 1
    \end{equation*}
Note that the random seeds $r_1, \ldots, r_m$ are independent.

    \item[$\bullet$] \textbf{Security:} Fix $i \in [m]$; we say that $\pi$ is \emph{$t$-private} if for every $T\subseteq [s]$ with $|T| \leq t$ and $x_i, x'_i \in \mathcal{X}$, $\Share(x_i)|_T$ has the same distribution as $\Share(x'_i)|_T$, over the randomness $\mathbf{r} \in \mathcal{R}^m$ used in $\Share$.

\end{itemize}

\end{definition}

\begin{remark}
    We remark that in the definition of HSS, the reconstruction algorithm $\Rec$ does \emph{not} need to know the identity of the function $f$ being computed, while the $\Eval$ function does.  In some contexts it makes sense to consider an HSS scheme for $\mathcal{F} = \{f\}$, in which case $f$ is fixed and known to all.  Our results in this work apply for general collections $\mathcal{F}$ of low-degree, multivariate polynomials, and in particular cover both situations.
\end{remark}

We will focus on \emph{linear} HSS schemes, which means that both $\Share$ and $\Rec$ are $\F$-linear over some finite field $\F$; we never require $\Eval$ to be linear.  More precisely, we have the following definition.
\begin{definition}[Linear HSS]
    Let $\mathbb{F}$ be a finite field.
    \begin{itemize}
        \item[$\bullet$] We say that an $s$-server HSS $\pi = (\Share,\Eval,\Rec)$ has \emph{linear reconstruction} if:
        \begin{itemize}
            \item  $\mathcal{Z} = \F$, so each output share $z_i \in \F^{n_i}$ is a vector over $\F$;
            \item $\mathcal{O} = \F^o$ is a vector space over $\F$; and
            \item $\Rec: \F^{\sum_i n_i} \to \F^o$ is $\F$-linear.
        \end{itemize}

        \item[$\bullet$] We say that $\pi$ has \emph{linear sharing} if $\mathcal{X}$, $\mathcal{R}$, and $\mathcal{Y}$ are all $\F$-vector spaces, and $\Share$ is $\F$-linear.

        \item[$\bullet$] We say that $\pi$ is \emph{linear} if it has both linear reconstruction and linear sharing. Note there is no requirement for $\Eval$ to be $\mathbb{F}$-linear.
    \end{itemize}
\end{definition}

Our main focus will be on the \emph{download rate} of linear HSS schemes.
\begin{definition}[Download cost, dowload rate]\label{def:downloadrate}
    Let $s, t$ be integers and let $\mathcal{F}$ be a class of functions with input space $\mathcal{X}^m$ and output space $\mathcal{O}$. Let $\pi$ be a $s$-server $t$-private HSS for $\mathcal{F}$. Let $z_i \in \mathcal{Z}^{n_i}$ for $i \in [s]$ denote the output shares.

    \begin{itemize}
        \item[$\bullet$] The \emph{download cost} of $\pi$ is given by

        \begin{equation*}
            \textsf{DownloadCost}(\pi):=\sum_{i \in [s]} \|{z_i}\|,
        \end{equation*}
        where we recall that $\|{z_i}\| = n_i \log_2|\mathcal{Z}|$ denotes the number of bits used to represent $z_i$. 
        \item[$\bullet$] The \emph{download rate} of $\pi$ is given by 
        \begin{equation*}
            \textsf{DownloadRate}(\pi):=\frac{\log_2| \mathcal{O}|}{\textsf{DownloadCost}(\pi)}.
        \end{equation*}
        
    \end{itemize}
\end{definition}

Thus, the download rate is a number between $0$ and $1$, and we would like it to be as close to $1$ as possible.

\subsection{Polynomial Function Classes}

Throughout, we will be interested in classes of functions $\mathcal{F}$ comprised of low-degree polynomials.
\begin{definition}
    Let $m > 0$ be an integer and $\mathbb{F}$ be a finite field. We define
    \begin{equation*}
        \POLYdmF := \lbrace f \in \mathbb{F}[X_1, \ldots, X_m] : \deg(f) \leq d \rbrace
    \end{equation*}
    to be the class of all $m$-variate polynomials of degree at most $d$, with coefficients in $\mathbb{F}$. 
\end{definition}
We are primarily interested in \emph{amortizing} the HSS computation over $\ell$ instances of $\POLYdmF$.  In this case, we will take our function class to be (a subset of) $\POLYdmF^\ell$ for some $\ell \in \mathbb{Z}^+$. 
\begin{remark}[Amortization over the computation of $\ell$ polynomials]\label{rem:amortization}
    If $\mathcal{F}$ is (a subset of) $\POLYdmF^\ell$, then the definition of HSS can be interpreted as follows: 
\begin{itemize}
\item     There are $\ell \cdot m$ input secrets, $x_k^{(i)}$ for $i \in [\ell]$ and $k \in [m]$, and each are shared independently among the $s$ servers.  (That is, in Definition~\ref{def:HSS}, we take $m \gets \ell \cdot m$). 

    \item There are $\ell$ target functions $f_1, \ldots, f_\ell \in \POLYdmF$, and the goal is to compute $f_i(x_1^{(i)}, \ldots, x_m^{(i)})$ for each $i \in [\ell]$.  (That is, in Definition~\ref{def:HSS}, $\mathbf{f} = (f_1, \ldots, f_\ell)$ is an element of $\mathcal{F} = \POLYdmF^\ell$).

    \item Each server $j \in [s]$ sends a function $z_j$ of all of their output shares, which importantly can combine information from across the $\ell$ instances; then the output client reconstructs $f_i(x_1^{(i)}, \ldots, x_m^{(i)})$ for each $i \in [\ell]$ from these shares.
    \end{itemize}

    In particular, we remark that this notion of amortization is interesting even for $d = m = 1$, when $f_1 = f_2 = \cdots = f_\ell$ is identity function.  In \cite{FIKW22}, that problem was called \emph{HSS for concatenation.}
\end{remark}

As hinted at above, our infeasibility results hold not just for $\POLYdmF^\ell$ but also for any $\mathcal{F} \subseteq \POLYdmF$ that contains monomials with at least $d$ different variables.
\begin{definition}\label{def:nontriv}
    Let $\mathcal{F} \subseteq \POLYdmF^\ell$.  We say that $\mathcal{F}$ is \emph{non-trivial} if there exists some $\mathbf{f} = (f_1, \ldots, f_\ell) \in \mathcal{F}$ so that for all $i \in [\ell]$, $f_i$ contains a monomial with at least $d$ distinct variables.
\end{definition}
(We note that our \emph{feasibility} results are stated in terms of $\mathcal{F} = \POLYdmF^\ell$; trivially these extend to any subset $\mathcal{F}\subseteq \POLYdmF^\ell$).

As mentioned in the introduction (Theorem~\ref{thm:FIKWinformal}), the work \cite{FIKW22} showed that any linear HSS scheme for $\POLYdmF^\ell$ (for any $\ell$) can have download rate at most $(s-dt)/s$:
We recall the following theorem from \cite{FIKW22}.
\begin{theorem}[\cite{FIKW22}]\label{thm:FIKW rate LB}
    Let $t,s,d,m,\ell$ be positive integers so that $m \geq d$. Let $\mathbb{F}$ be any finite field and $\pi$ be a $t$-private $s$-server linear HSS scheme for $\POLYdmF^\ell$. Then $dt< s$, and $\textsf{DownloadRate}(\pi) \leq (s - dt)/s$.
\end{theorem}

\subsection{CNF Sharing}
The main $\Share$ function that we consider in this work is \emph{CNF sharing}~\cite{ISN89}.

\begin{definition}[$t$-private CNF sharing] \label{def:CNF}
Let $\F$ be a finite field.
The $t$-private, $s$-server CNF secret-sharing scheme over $\F$ is a function $\Share: \F \times \F^{\binom{s}{t} - 1} \to \left(\F^{\binom{s-1}{t}}\right)^s$ that shares a secret $x \in \F$ as $s$ shares $y_j \in \F^{\binom{s-1}{t}}$, using $\binom{s}{t} - 1$ random field elements, as follows.

Let $x \in \F$, and let $\mathbf{r} \in \F^{\binom{s}{t} - 1}$ be a uniformly random vector.  Using $\mathbf{r}$, choose $y_T \in \F$ for each set $T \subseteq [s]$ of size $t$, as follows: The $y_T$ are uniformly random subject to the equation
\[ x = \sum_{T \subseteq [s]: |T| = t} y_T.\]
Then for all $j \in [s]$, define
\[ \Share(x,\mathbf{r})_j = ( y_T \,:\, j \not\in T ) \in \F^{\binom{s-1}{t} }.\]
\end{definition}

We observe that CNF-sharing is indeed $t$-private.  Any $t+1$ servers between them hold all of the shares $y_T$, and thus can reconstruct $x = \sum_T y_T$.  In contrast, any $t$ of the servers (say given by some set $S \subseteq [s]$) are missing the share $y_S$, and thus cannot learn anything about $x$.

The main reason we focus on CNF sharing is that it is \emph{universal} for linear secret sharing schemes:

\begin{theorem}[\cite{CDI05}]\label{thm:CNFuniversal}
    Suppose that $x \in \F$ is $t$-CNF-shared among $s$ servers, so that server $j$ holds $y_j \in \F^{\binom{s-1}{t}}$, and let $\Share'$ be any other linear secret-sharing scheme for $s$ servers that is (at least) $t$-private.  Then the shares $y_j$ are locally convertible into shares of $\Share'$.  That, is there are functions $\phi_1, \ldots, \phi_s$ so that
    $(\phi_1(y_1), \ldots, \phi_s(y_s))$
    has the same distribution as
    $\Share'(x, \mathbf{r})$
    for a uniformly random vector $\mathbf{r}$.
\end{theorem}
In particular, we prove several results of the form ``no linear HSS with CNF sharing can do better than 
\_\_\_\_.'' Because of Theorem~\ref{thm:CNFuniversal}, these results imply that ``no linear HSS with \emph{any} linear sharing scheme can do better than \_\_\_\_.''

Finally, we record a useful lemma about HSS with $t$-CNF sharing, which says that in order to recover a degree-$d$ monomial, the output client must contact at least $dt + 1$ servers. This lemma is implicit in Lemma 2 of \cite{FIKW22} 
, but for clarity we give a statement that is more directly applicable to the viewpoint of our work and provide a proof.

\begin{lemma}\cite{FIKW22}\label{lem: query spread}
Fix $s,d,t$ so that $s \geq dt + 1$, and suppose $m \geq d$.
Let $\pi = (\Share, \Eval, \Rec)$ be any linear HSS that $t$-CNF shares $m$ secrets $x_1, \ldots, x_m \in \F$ among $s$ servers, for a nontrivial function class $\mathcal{F} \subset \POLYdmF$.
Then $\Rec$ must depend on output shares $z_j$ from at least $dt + 1$ distinct servers $j \in [s]$.

\end{lemma}

\begin{proof}
From the definition of a nontrivial function class (Definition~\ref{def:nontriv}, with $\ell=1$), there is some $\hat{f} \in \mathcal{F}$ so that $\hat{f}$ contains a monomial with at least $d$ distinct variables.
    Up to relabeling the secrets, we may assume that $\hat{f}$ contains the monomial $f(x_1, \ldots, x_m) := \prod_{i=1}^d x_i$.  As the monomials of $\hat{f}$ are linearly independent over $\F[X_1, \ldots, X_m]$ and $\Rec$ is linear, if $\Rec$ could recover $\hat{f}$ with probability $1$, then $\Rec$ could recover $f$ with probability $1$ as well.  Thus, it suffices to show that no linear $\Rec$ with output shares from at most $dt$ servers can compute $f$ with probability $1$. 
    
    Let $J \subseteq [s]$ with $\abs{J} < dt+1$ be any collection of at most $dt$ servers, and suppose towards a contradiction that the function $\Rec$ depends only on the output shares $(z_j : j \in J)$.
    
    For $j \in J$, let
    \begin{equation*}
        \mathbf{y}_{j} = \left( y_{h,T}\; : \; h \in [d], T \subseteq [s], |T| = t, j \not\in T \right)
    \end{equation*}
    denote the set of $t$-CNF shares of $x_1, \ldots, x_d$ held by server $j$. Let
    $g_{j}$ be the polynomial over $\F$ in the variables $\mathbf{y}_{j}$ so that
    \[ g_{j}(\mathbf{y}_{j}) = \Eval(f, j, \mathbf{y}_{j}).\]
    (Note that over a finite field $\F$, any function can be represented as a polynomial, whose degree in each variable is at most $|\F|-1$).
   
    Let $T_1^*, \ldots, T_d^* \subseteq [s]$ be any collection of $t$-sized subsets satisfying $J \subseteq \bigcup_{i \in [d]} T_i^*$. 
    This exists because $\abs{J} \leq dt$, so we can choose $d$ size-$t$ subsets that cover $J$. 
    Next we observe that the monomial 
    \[ W = \prod_{i=1}^d y_{i,T_i^*} \]  cannot appear in the polynomial $g_{j}$ for any $j \in J$, as none of the servers $j \in J$ can compute $W$.  

    Because the HSS scheme is linear and we are assuming that $\Rec$ depends only on $(z_j : j \in J)$, the output of $\Rec$ can be written as a linear combination of the polynomials $g_j$ for $j \in J$:
    \[ \Rec(z_j :j \in J) = \sum_{j \in J} \alpha_j g_j(\mathbf{y}_j), \]
    for some $\alpha_j \in \F$.
    In particular, $W$ cannot appear in $\Rec(z_j : j \in J)$ when viewed as a formal polynomial.
    
However, as a formal polynomial, the function that we would like to reconstruct is 
\[ f(x_1, \ldots, x_d) = \prod_{i \in [d]} x_i = \prod_{i \in [d]} \left( \sum_T y_{i,T_i} \right) = \sum_{T_1, \ldots, T_d} \prod_{i=1}^d y_{i,T_i}, \]
which \emph{does} include the monomial $W$.  (Above, the sum on the right hand side is over all $(T_1, \ldots, T_d)$ where each $T_i \subseteq [s]$ has size $t$). 

In particular, there are values of the shares $y_{i,T_i}$ so that $\Rec(z_j : j \in J)$ is not the same as $f(x_1, \ldots, x_d)$, which means that the probability that they are the same is less than 1, contradicting the correctness requirement of an HSS scheme.

\end{proof}

\subsection{Linear Codes and Labelweights}

Throughout, we will be working with \emph{linear codes} $\cC \subset \F^n$.  Such a code is just a subspace of $\F^n$.  For a linear code $\cC \subseteq \F^n$ of dimension $\ell$, we say that a matrix $G \in \F^{\ell \times n}$ is a \emph{generator matrix} for $\cC$ if $\cC = \mathrm{rowSpan}(G)$.  Note that generator matrices are not unique.  If $G$ has the form
\[ G = [I | A ]\]
where $I = \F^{\ell \times \ell}$ is the identity matrix and $A \in \F^{\ell \times (n-\ell)}$, we say that $G$ is in \emph{systematic form}, and we refer to the first $\ell$ coordinates as \emph{systematic coordinates}.  
The other coordinates we refer to as \emph{non-systematic coordinates}.
The \emph{rate} of a linear code $\cC \subset \F^n$ of dimension $\ell$ is defined as
\[ \textsf{Rate}(\cC) := \frac{\ell}{n}. \]

As discussed in the Introduction, our characterization of download-optimal linear HSS schemes is in terms of linear codes with good \emph{labelweight}.

\begin{definition}[Labeling Function]
    Let $U, V$ be finite domains satisfying $|U| \geq |V|$. We say that a mapping $\mathcal{L}: U \to V$ is a \textit{labeling function} (or simply \textit{labeling}) of $U$ by $V$ if $\mathcal{L}$ is a surjection.
\end{definition}

\noindent
We will usually take $U = [n]$ and $V = [s]$ for $n, s \in \mathbb{Z}^+$ with $n \geq s$.
For a code $\cC \subseteq \F^n$, and a labeling function $\mathcal{L}:[n] \to [s]$, the labelweight of a codeword is the number of distinct labels that its support touches, and the labelweight of a code is the minimum labelweight of any nonzero codeword.  More precisely, we have the following definition.

\begin{definition}[Labelweight]\label{def:labelweight}
Let $\mathcal{L}$ be a labeling of $[n]$ by $[s]$ where $n,s \in \mathbb{Z}^+$ with $n \geq s$. Let $\mathbb{F}$ denote a field. Given $\mathbf{c} \in \mathbb{F}^n$, we define the \textit{labelweight} of $\mathbf{c}$ by
\begin{equation*}
    \Delta_\mathcal{L} (\mathbf{c}) := \left|\left\lbrace \mathcal{L}(i) \; : \; \mathbf{c}_i \neq 0 \right\rbrace\right|.
\end{equation*}

\noindent
For a code $\mathcal{C} \subseteq \mathbb{F}^n$, we define \textit{labelweight} of $\mathcal{C}$ by
\begin{equation*}
    \Delta_\mathcal{L} (\mathcal{C}) := \min_{\mathbf{0} \neq \mathbf{c} \in \mathcal{C}} \Delta_\mathcal{L}(\mathbf{c}).
\end{equation*}

\end{definition}

We note that labelweight is a generalization of Hamming weight; indeed, let $\iota: [n] \to [n]$ denote the identity function on $[n]$, which may be viewed as a labeling of $[n]$ by $[n]$. Given a linear code $\mathcal{C}$ of length $n$ and $\mathbf{c} \in \mathcal{C}$, we see that $\Delta_\iota (\mathbf{c})$ and $\Delta_\iota(\mathcal{C})$ are equivalent to the Hamming weight of $\mathbf{c}$ and the minimum Hamming distance of $\mathcal{C}$, respectively.

\section{Linear HSS Schemes and Good Labelweight Codes}\label{sec:equiv}
In this section we show that finding optimal-download linear HSS schemes for low-degree multivariate polynomials is equivalent to finding linear codes with high labelweight. 
Before we show the equivalence, we make a few observations about linear reconstruction algorithms.  The first is just the observation that any linear reconstruction scheme can be regarded as matrix:

\begin{observation}\label{obs: linear reconstruction is matrix multiplication}
    Let $\ell, t, s, d, m, n$ be integers. Let $\pi = (\Share, \Eval, \Rec)$ be a $t$-private, $s$-server HSS for some function class $\mathcal{F} \subseteq \POLYdmF^\ell$ with linear reconstruction $\Rec: \F^n \to \F^\ell$, where $n = \sum_{j \in [s]} n_j$, and the output share $z_j$ of server $j$ is an element of $\F^{n_j}$.
    Let $\mathbf{z} \in \F^{n}$ be the vector of all output shares.  That is,
    \[ \mathbf{z} = z_1 \circ z_2 \circ \cdots \circ z_s, \]
    where $\circ$ denotes concatenation.

    Then there exists a matrix $G_\pi \in \mathbb{F}^{\ell \times n}$ so that, for all $f \in \mathcal{F}$ and for all secrets $\mathbf{x} \in (\F^m)^\ell$, 
    \begin{equation*}
        \Rec(\mathbf{z}) = G_\pi \mathbf{z} = f(\mathbf{x})= \left[\begin{array}{c}
             f_1( \mathbf{x}^{(1)})  \\
             f_2( \mathbf{x}^{(2)}) \\
             \vdots\\
             f_\ell( \mathbf{x}^{(\ell)}) 
        \end{array} \right]
    \end{equation*}
\end{observation}
For a linear HSS $\pi$, we call $G_\pi$ as in the observation above the \emph{reconstruction matrix} corresponding to $\Rec$.

We next observe that $G_{\pi}$ has full rank.
\begin{lemma}\label{lem:Gpi full rank}
    Let $t,s,d,m,\ell$ be positive integers so that $m \geq d$ and $n \geq \ell$, and let $\pi$ be a $t$-private $s$-server linear HSS for some $\mathcal{F} \subseteq \POLYdmF$, so that $\mathcal{F}$ contains an element $(f_1, \ldots, f_\ell)$ where for each $i \in [\ell]$, $f_i$ is non-constant.  Then $G_\pi \in \F^{\ell \times n}$ has rank $\ell$.
\end{lemma}
\begin{proof}
    Suppose that $G_\pi$ does not have rank $\ell$.  Then there is a left kernel vector $\mathbf{v}$, so that $\mathbf{v}^T G_\pi = 0$.  Let $(f_1, \ldots, f_\ell) \in \mathcal{F}$ be as in the lemma statement. By Observation~\ref{obs: linear reconstruction is matrix multiplication}, we then have
    \[ 0 = \mathbf{v}^T G_\pi \mathbf{z} = \mathbf{v}^T \begin{bmatrix} f_1(\mathbf{x}^{(1)}) \\ \vdots \\ f_\ell(\mathbf{x}^{(\ell)}) \end{bmatrix}\] 
    for all values of  secrets $\mathbf{x}^{(1)}, \ldots, \mathbf{x}^{(\ell)}$.  In particular, 
    $\sum_{i=1}^\ell v_i f_i(\mathbf{x}^{(i)})$ is identically zero as a polynomial in the variables $x_k^{(i)}$ for $k \in [m], i \in [\ell]$.
    However, this is a contradiction because the variables that show up in $\mathbf{x}^{(i)}$ are different for different values of $i \in [\ell]$, and thus cannot cancel with each other.  
\end{proof}

\subsection{Optimal-Download HSS Implies Good Labelweight Codes}\label{sec:HSS_to_lblwt}
We now show the forward direction of the equivalence.

\begin{lemma}\label{lem:HSS_implies_lblwt}
    Let $\ell, t, s, d, m$ be integers, with $m \geq d$. Suppose there exists a $t$-private, $s$-server HSS $\pi = \left( \Share, \Eval, \Rec \right)$ for some non-trivial (see Definition~\ref{def:nontriv}) $\mathcal{F} \subseteq \text{POLY}_{d,m}(\mathbb{F})^\ell$, with download rate $\textsf{DownloadRate}(\pi)=(s-dt)/s$. Let $n = \ell / \textsf{DownloadRate}(\pi)$.  Suppose also that $\pi$ is linear over $\F$, and the $\Share$ is $t$-CNF sharing.  Then there exists a linear code $\cC \subseteq \F^n$ with rate $\textsf{DownloadRate}(\pi)$ and a labeling $\cL:[n] \to [s]$ so that $\Delta_\cL(\cC) \geq dt + 1$.
    
\end{lemma}

\begin{proof}
    Let $G_\pi \in \mathbb{F}^{\ell \times n}$ be the matrix representation of the linear reconstruction algorithm $\Rec$, as in Observation~\ref{obs: linear reconstruction is matrix multiplication}.   Let $\mathcal{C}$ denote the linear code spanned by the rows of $G_\pi$. 
    By Lemma~\ref{lem:Gpi full rank},  $G_\pi$ has rank $\ell$, so the rate of $\cC$ is
    \[ \frac{\ell}{n} = \textsf{DownloadRate}(\pi).\]
    It remains only to show $\Delta_{\mathcal{L}}(\mathcal{C}) \geq dt+1$ for some labeling function $\cL$.  Recall from Observation~\ref{obs: linear reconstruction is matrix multiplication} that $n = \sum_{j \in [s]} n_j$ is the total number of symbols of $\F$ that appear in output shares across all of the servers, and that each column of $G$ is associated with one such symbol. Let $\cL:[n] \to [s]$ be the labeling function so that the $j$'th column of $G$ is associated with a symbol sent by the $\mathcal{L}(j)$'th server.

    Let $\mathbf{x} = \left( \mathbf{x}_1, \ldots, \mathbf{x}_\ell \right) \in \left(\mathbb{F}^m \right)^\ell$ be the vector of $m\ell$ secrets, each of which is independently $t$-CNF shared to the $s$ servers. Let $\mathbf{f} = (f_1, \ldots, f_\ell) \in \mathcal{F}$ be a function guaranteed by the fact that $\mathcal{F}$ is nontrivial (Definition~\ref{def:nontriv}).

    As in Observation~\ref{obs: linear reconstruction is matrix multiplication}, let $\mathbf{z} \in \mathbb{F}^n$ be the vector of $n$ output shares.  Notice that with the labeling function defined above, this means that the $j$'th coordinate of $\mathbf{z}$ is sent by server $\cL(j)$.

    Let $G_{\pi,1}, \ldots, G_{\pi,\ell}$ denote the rows of $G_\pi$ and observe that
    \begin{equation*}
        G_\pi \mathbf{z} = \left[ \begin{array}{c}
            G_{\pi,1}^T \mathbf{z}  \\
            G_{\pi, 2}^T \mathbf{z} \\
            \vdots\\
            G_{\pi, \ell}^T \mathbf{z}
        \end{array} \right] = \mathbf{f}(\mathbf{x}) = \left[ \begin{array}{c}
            f_1(\mathbf{x}^{(1)})  \\
             f_2(\mathbf{x}^{(2)}) \\
            \vdots\\
            f_\ell (\mathbf{x}^{(\ell)})
        \end{array} \right]
    \end{equation*}

    Let $\mathbf{v} \in \F^\ell \setminus \lbrace \mathbf{0} \rbrace$, and consider the codeword $\mathbf{c} \in \cC$ given by $\mathbf{c} = \mathbf{v}^TG_\pi$.  Then from the above, 
    \[ \mathbf{c}^T \mathbf{z} = \sum_{i=1}^\ell v_i f_i(\mathbf{x}^{(i)}) =: f(\mathbf{x}). \]
    In particular, for $\mathcal{F}' = \inset{ f }$, $\mathbf{c}$ gives a linear reconstruction algorithm $\Rec'$ for a linear HSS $\pi' = (\Share, \Eval, \Rec')$ for $\mathcal{F}'$.  Notice that $\mathcal{F}'$ is also non-trivial; indeed, each $f_i(\mathbf{x}^{(i)})$ contains distinct variables (so they cannot cancel), and each $f_i$ has a monomial with at least $d$ distinct variables; thus $f$ contains a monomial with at least $d$ distinct variables.
    
    Then by Lemma~\ref{lem: query spread}, $\Rec'$ must depend on output shares from at least $dt + 1$ distinct servers $j \in [s]$.  It follows from the definition of $\cL$ that $\Delta_\cL(\mathbf{c}) \geq dt + 1$.  Since $\mathbf{v}$, and hence $\mathbf{c}$, was arbitrary (non-zero), this implies that $\Delta_\cL(\cC) \geq dt + 1$, as desired.
\end{proof}

Lemma \ref{lem:HSS_implies_lblwt} shows that the existence of an optimal-rate linear HSS scheme for polynomials implies the existence of a linear code with with optimum labelweight.  The converse is also true, but before we prove that result, we need to dive a bit deeper into the structure of codes with good labelweight, which we do next.

\subsection{Structural Properties of Labeling Functions for Optimal Labelweight Codes}
We establish a few structural properties of the labeling functions $\cL$ that appear in codes with optimal labelweight.

Our main structural result states that for any code $\cC$ of rate $(s-dt)/s$, if $\cL$ is such that  $\Delta_\cL(\cC) \geq dt + 1$, then in fact the labeling function $\cL$ must be ``balanced,'' in the sense that each label shows up the same number of times.  Below and throughout the paper, for $\cL:[n] \to [s]$, we use $\cL^{-1}(\lambda) = \inset{ j \in [n] : \cL(j) = \lambda }$ to denote the preimage of $\lambda \in [s]$ under $\cL$.

\begin{lemma}\label{lem: dl is balanced}
    Let $\mathcal{C} \subseteq \mathbb{F}^n$ be a rate $(s-dt)/s$ linear code and $\mathcal{L}:[n] \to [s]$ a labeling of its coordinates by $[s]$ so that $\Delta_\mathcal{L}(\mathcal{C}) \geq dt+1$. Then for all $\lambda\neq \lambda' \in [s]$, $\abs{ \mathcal{L}^{-1}(\lambda)} = \abs{\mathcal{L}^{-1}(\lambda')}$. 
\end{lemma}

\begin{proof} 
    Since $\cC$ has rate $(s-dt)/s$, there exists some $j \in \mathbb{Q}^+$ such that $\mathcal{C}$ has block length $n = js$ and dimension $\ell = j(s-dt)$.  Let $G \in \mathbb{F}^{j(s-dt) \times js}$ be an arbitrary generator matrix of $\mathcal{C}$. The lemma holds if and only if for each $\lambda \in [s]$, we have $|\mathcal{L}^{-1}(\lambda)| = js/s = j$.  Thus, we assume towards a contradiction that this is not the case and break the proof into three cases, depending on the relationship between $s$ and $2dt$.

    \noindent
    \textbf{Case 1: $s=2dt$.} In this case, consider $\lambda_1, \ldots, \lambda_{dt}$ chosen so that $\sum_{i=1}^{dt} |\mathcal{L}^{-1}(\lambda_i)|$ is minimized. As we assume towards a contradiction that there exists $\lambda \in [s]$ such that $\abs{\mathcal{L}^{-1}(\lambda)} \neq j$, there exists a choice of $\lambda_1, \lambda_2, \ldots, \lambda_{dt} \in [s]$ so that 
    \begin{equation*}
        \left|\bigcup_{i=1}^{dt} \mathcal{L}^{-1}(\lambda_i)\right| = jdt-\sigma < jdt
    \end{equation*}
    for some $\sigma >0$. Without loss of generality we may assume that $\lambda_1, \ldots, \lambda_{dt}$ label the final $jdt-\sigma$ columns of $G$. Since $s=2dt$, $G$ has $j(s-dt)=jdt$ rows; hence there exists some nonzero $\mathbf{m} \in \mathbb{F}^{jdt}$ such that $\mathbf{m}^T G$ has no support in its final $jdt-\sigma$ coordinates. But then $\Delta_\mathcal{L}(\mathbf{m}^T G) \leq dt$, contradicting $\Delta_\mathcal{L}(\mathcal{C})\geq dt+1$.

    \noindent
    \textbf{Case 2: $s\geq 2dt+1$.} In this case, consider $\lambda_1, \ldots, \lambda_{dt}$ so that $\sum_{i=1}^{dt} |\mathcal{L}^{-1}(\lambda_i)|$ is maximized. As we assume towards a contradiction that there exists $\lambda \in [s]$ such that $\abs{\mathcal{L}^{-1}(\lambda)} \neq j$, there exists a choice of $\lambda_1, \lambda_2, \ldots, \lambda_{dt} \in [s]$ so that
    \begin{equation*}
        \left|\bigcup_{i=1}^{dt} \mathcal{L}^{-1}(\lambda_i)\right| = jdt + \sigma > jdt
    \end{equation*}
    for some $\sigma > 0$. Without loss of generality we may assume that $\lambda_1, \ldots, \lambda_{dt}$ label the final $jdt+\sigma$ columns of $G$, leaving the first $js-(jdt+\sigma) = j(s-dt) -\sigma$ columns labeled with $[s] \setminus \lbrace \lambda_1, \ldots, \lambda_{dt} \rbrace$. Since $G$ has precisely $j(s-dt)$ rows, there exists some $\mathbf{m} \in \mathbb{F}^{jdt}$ such that $\mathbf{m}^TG$ has no support in its first $j(s-dt) -\sigma$ coordinates. It follows that $\Delta_\mathcal{L}(\mathbf{m}^TG) \leq dt$, contradicting $\Delta_\mathcal{L}(\mathcal{C})\geq dt+1$.

    \noindent
    \textbf{Case 3: $s \leq 2dt - 1$.} Note that $s \geq dt + 1$, as we assume that $\Delta_{\mathcal{L}}(\cC) \geq dt + 1$.  Thus, we may write $s = dt + \hat{s}$, where $\hat{s} \in [dt-1]$. In this case, consider $\lambda_1, \ldots, \lambda_{\hat{s}}$ so that $\sum_{i=1}^{\hat{s}} |\mathcal{L}^{-1}(\lambda_i)|$ is minimized. As we assume towards a contradiction that there exists $\lambda \in [s]$ such that $\abs{\mathcal{L}^{-1}(\lambda)} \neq j$, there exists a choice of $\lambda_1, \lambda_2, \ldots, \lambda_{\hat{s}} \in [s]$ so that
    \begin{equation*}
        \left|\bigcup_{i=1}^{\hat{s}} \mathcal{L}^{-1}(\lambda_i) \right| = j\hat{s} - \sigma < j\hat{s}
    \end{equation*}
    for some $\sigma > 0$. Without loss of generality we may assume that $\lambda_1, \ldots, \lambda_{\hat{s}}$ label the first $j\hat{s} - \sigma$ columns of $G$. Since there are $j\hat{s}$ rows of $G$, there exists some $\mathbf{m} \in \mathbb{F}^{j\hat{s}}$ such that $\mathbf{m}^TG$ has no support in its first $j\hat{s} -\sigma$ coordinates. It follows that $\Delta_\mathcal{L}(\mathbf{m}^TG) \leq dt$, contradicting $\Delta_\mathcal{L}(\mathcal{C})\geq dt+1$.
\end{proof}

\begin{corollary}\label{cor: j is an int, labeling is wlog}
     Let $s,d,t \in \mathbb{Z}^+$ satisfying $s-dt > 0$. Let $\mathcal{C}$ be a rate $(s-dt)/s$ linear code and $\mathcal{L}$ a labeling of its coordinates by $[s]$ so that $\Delta_\mathcal{L}(\mathcal{C}) \geq dt+1$. Then there exists some $j \in \mathbb{Z}^+$ such that:
     \begin{itemize}
         \item[(i)] $\mathcal{C}$ has length $n = js$ and dimension $\ell = j(s-dt)$; 
         \item[(ii)] $|\mathcal{L}^{-1}(\lambda)| = j $ for all $\lambda \in [s]$; 
         \item[(iii)] there is a re-ordering of the coordinates of $\cC$ so that $\mathcal{L}:[js] \to [s]$ is given by $\mathcal{L}:x \mapsto \lceil x/j \rceil$.
     \end{itemize}
\end{corollary}

\begin{proof}
Define $j = n/s$ so that $js$ is equal to the block length of $\cC$.  Then Lemma~\ref{lem: dl is balanced} implies that $|\mathcal{L}^{-1}(\lambda)| = j$ for all $\lambda \in [s]$, and in particular that $j$ is an integer.  This establishes all three points.
\end{proof}

\begin{corollary}\label{cor: generator of good labelweight code is MDS like}
Let $j$ be a positive integer, and 
    let $\mathcal{C}\subseteq \mathbb{F}_q^{js}$ be a linear code of length $js$ and  dimension $j(s-dt)$.  Let $\mathcal{L}:[js] \to [s]$ be a labeling such that $\Delta_\mathcal{L}(\mathcal{C}) \geq dt+1$. Let $G \in \mathbb{F}_q^{j(s-dt)\times js}$ be an arbitrary generator matrix for $\mathcal{C}$.  Given $\Lambda \subseteq [s]$, let $G(\Lambda)$ be the submatrix of $G$ consisting of the columns $G_i$ of $G$ for all $i \in [js]$ so that $\mathcal{L}(i) \in \Lambda$.
    Then for any $\Lambda  \subseteq [s]$ with $|\Lambda| = s - dt$,  $G(\Lambda) \in \mathbb{F}^{j(s-dt) \times j(s-dt)}$ and $\det(G(\Lambda)) \neq 0$.
\end{corollary}

\begin{proof}
     The dimensions of $G$ follow immediately from Corollary \ref{cor: j is an int, labeling is wlog}; it remains only to show that $G(\Lambda)$ is non-singular. Assume towards a contradiction there exists $\mathbf{m} \in \mathbb{F}_2^{j(s-dt)}$ such that $\mathbf{m}^TG(\Lambda) = \mathbf{0}$.
     But then the support of $\mathbf{m}^T G \in \cC$ does not include any coordinates labeled with any $\lambda \in \Lambda$.  In particular, $\Delta_\cL(\mathbf{m}^T G) \leq dt$, which contradicts $\Delta_\cL(\cC) \geq dt + 1$.
\end{proof}

\subsection{Good Labelweight Codes Imply Optimal-Download HSS}\label{sec:lblwt_to_HSS}

Now, we can prove a converse to Lemma~\ref{lem:HSS_implies_lblwt}.  

\begin{theorem}\label{thm:lblwtImpliesHSS}
    Let $\ell, t, s, d, m, n$ be integers, with $m \geq d$ and $\ell s = n(s-dt)$. 
    Suppose that there exists a linear code $\cC \subseteq \F^n$ and a labeling $\cL: [n] \to [s]$ so that $\Delta_\cL(\cC) \geq dt + 1$.
    Then there exists a $t$-private, $s$-server linear HSS $\pi = (\Share, \Eval, \Rec)$ for $\POLYdmF^\ell$ with download rate $\textsf{DownloadRate}(\pi) \geq (s - dt)/s$.
\end{theorem}

\begin{proof}
    The proof is by construction; we give an example of this construction in Appendix \ref{sec: example (apx)}, so the reader may wish to flip to that as they read the proof. 
    
    Let $\cC$ be as in the theorem statement, and let $G \in \F^{\ell \times n}$ be any generator matrix for $\cC$.
    In order to define $\pi$, we need to define the functions $\Share$, $\Eval$, and $\Rec$.  For $\Share$, we will use $t$-CNF sharing (Definition~\ref{def:CNF}).  We define $\Rec$ using the generator matrix $G$.  In particular, we will (soon) define $\Eval$ so that server $j \in [s]$ returns $n_j := |\cL^{-1}(j)|$ elements of $\F$ as output shares.  We will gather these output shares into a vector $\mathbf{z} \in \F^n$, where $n = \sum_{j \in [s]} n_j$.  Then $\Rec$ will be given by
    \[ \Rec(\mathbf{z}) = G\mathbf{z}.\]
    
    Finally, it remains to define $\Eval$.  To do so, we will set up a linear system that essentially says ``the $\Rec$ function we just defined is correct.''  Then we will show that this linear system has a solution, and that will yield our $\Eval$ function.

    Below, we assume without loss of generality that the function $\mathbf{f} = (f_1, \ldots, f_\ell)$ that the HSS scheme is trying to compute has $f_j(x_1, \ldots, x_m) = \prod_{i=1}^d x_i$ for all $j \in [\ell]$, and we will define $\Eval(\mathbf{f}, j, \mathbf{y}_j)$ on only this $\mathbf{f}$.  To obtain the general result for any polynomials $(p_1, p_2, \ldots, p_\ell) \in \POLYdmF^\ell$, we first observe that the argument goes through if the $f_j$'s are \emph{any} monomials of degree at most $d$, possibly with a leading coefficient; this follows by re-ordering the secrets, and possibly including some dummy secrets that are identically equal to a constant (that is, increasing the parameter $m$, which does not appear in any of the results) to obtain monomials of degree less than $d$ and/or with a leading coefficient. 
    Then to pass to general polynomials, and not just monomials, we observe that since $\Rec$ is linear, we may define $\Eval$ additively; that is, if $p_i(\mathbf{x}^{(i)}) = \sum_r f_{i,r}(\mathbf{x}^{(i)})$ for some monomials $f_{i,r}$, server $j$ will compute $\Eval$ on each $\mathbf{f}_r = (f_{1,r}, \ldots, f_{\ell,r})$, and then return their sum.

    Now we return to the task of setting up a linear system to define $\Eval$ for the particular function $\mathbf{f}$ defined above.
    In order to set up this linear system, we introduce some notation.  Let 
    $\mathcal{T} = \lbrace T \subseteq [s] : |T| = t \rbrace$  be the set of size-$t$ subsets of $[s]$.  Let $\mathbf{x} = (\mathbf{x}^{(1)}, \ldots, \mathbf{x}^{(\ell)}) \in (\F^m)^\ell$ denote the secrets to be shared.  For $T \in \mathcal{T}$, $r \in [\ell]$, and $j \in [m]$, let $y^{(r)}_{j, T}$ denote the CNF shares of $x^{(r)}_j$, so 
    \[ x^{(i)}_j = \sum_T y^{(i)}_{j,T}.\]  Thus, for each $i \in [\ell]$ the function we would like to recover is
    \begin{equation}\label{eq:want}
    f_i(\mathbf{x}^{(i)}) = \sum_{\mathbf{T} \in \mathcal{T}^d} \prod_{k=1}^d y^{(i)}_{k,T_k}. 
    \end{equation}
    Let $\mathbf{y}_j$ denote the set of CNF shares that server $j$ holds: $\mathbf{y}_j = (y_{k,T}^{(i)} : k \in [d], T \in \mathcal{T}, i \in [\ell], j \not\in T)$.  We will treat $\mathbf{y}_j$ as tuples of formal variables.

Next, we define the following classes of monomials, in the variables $y_{j,T}^{(i)}$.  Let
\begin{equation*}
    \mathcal{M} = \left\lbrace y_{1,T_1}^{(i)} y_{2, T_2}^{(i)} \cdots y_{d, T_d}^{(i)} \; : \; \mathbf{T} \in \mathcal{T}^d, i \in [\ell] \right\rbrace.
\end{equation*}
Given a server $j \in [s]$, let
\begin{equation*}
    \mathcal{M}_j  = \left\lbrace y_{1,T_1}^{(i)} y_{2, T_2}^{(i)} \cdots y_{d, T_d}^{(i)} \in \mathcal{M}  \; : \; \mathbf{T} \in \mathcal{T}^d, i \in [\ell], j \not\in \bigcup_{k \in [d]} T_k \right\rbrace.
\end{equation*}
That is, $\mathcal{M}_j$ is
the subset of $\mathcal{M}$ locally computable by server $j$.  

The function $\Eval(\mathbf{f}, j, \mathbf{y}_j)$ that determines server $j$'s output shares will be defined by a sequence of $n_j = |\cL^{-1}(j)|$ polynomials of degree $d$, constructed from the monomials in $\mathcal{M}_j$.  To that end, we will define a vector of variables $\mathbf{e} \in \F^{ \sum_{r \in [n]} |\mathcal{M}_{\cL(r)}|}$, indexed by pairs $(r, \chi)$ for $\chi \in \cM_{\cL(r)}$.  The vector $\mathbf{e}$ will encode the function $\Eval$ as follows.  For each $r \in [n]$,  we define $z_r = z_r(\mathbf{y}_{\cL(r)})$ to be the polynomial in the variables $\mathbf{y}_{\cL(r)}$ given by
\begin{equation}\label{eq:zr}
z_r(\mathbf{y}_{\cL(r)})  := \sum_{\chi \in \cM_{\cL(r)}} \mathbf{e}_{r, \chi} \cdot \chi(\mathbf{y}_{\cL(r)}). 
\end{equation}
Then we define $\Eval$ by
\[ \Eval( \mathbf{f}, j, \mathbf{y}_j ) = ( z_r(\mathbf{y}_j) : r \in \cL^{-1}(j) ) \in \F^{n_j} \]
for each server $j \in [s]$.

Now we will set up our system to solve for the coefficients in $\mathbf{e}$, which will define $\Eval$ as above.

Define the matrix $S \in \F^{\ell |\cM| \times \sum_{r \in [n]}|\cM_{\cL(r)}|}$ as follows.
\begin{itemize}
    \item The rows of $S$ are indexed by pairs $(i, \fm) \in [\ell] \times \cM$.
    \item The columns of $S$ are indexed by pairs $(r, \chi)$ for $r \in [s]$ and $\chi \in \cM_r$.
    \item The entry of $S$ indexed by $(i,\fm)$ and $(r, \chi)$ is given by:
    \[ S[(i,\fm), (r, \chi)] = \begin{cases} G[i,r] & \fm = \chi \\ 0 & \text{else} \end{cases}.\]
\end{itemize}

Define a vector $\mathbf{g} \in \F^{\ell|\cM|}$ so that the coordinates of $\mathbf{g}$ are indexed by pairs $(i,\fm) \in [\ell] \times \cM$, so that
\[ \mathbf{g}[(i,\fm)] = \begin{cases} 1 & \psi_i(\fm) \\ 
0 & \text{else} \end{cases}\]
where
\[ \psi_i(\fm) = \begin{cases} 1 & \text{ $\fm$ is of the form $\prod_{k=1}^d y_{k,T_k}^{(i)}$ for some $\mathbf{T} \in \mathcal{T}^d$ } \\ 0 & \text{else} \end{cases}\]
Notice that \eqref{eq:want} implies that for $i \in [\ell]$,
\begin{equation}\label{eq:want2} f_i(\mathbf{x}^{(i)}) = \sum_{\mathbf{T} \in \mathcal{T}^d} \prod_{k=1}^d y_{k,T_k}^{(i)} = \sum_{\fm \in \mathcal{M}} \psi_i(\fm) \cdot \fm(\mathbf{y}).\end{equation}

\begin{claim}\label{claim:correct}
    Suppose that $S \cdot \mathbf{e} = \mathbf{g}$.  Then the HSS scheme $\pi = (\Share, \Eval, \Rec)$, where $\Share$ and $\Rec$ are as above, and $\Eval$ is defined by $\mathbf{e}$ as above, is correct.  (That is, it satisfies the \emph{Correctness} property in Definition~\ref{def:HSS}).
\end{claim}
\begin{proof}
    Recall that we will set $\Rec(\mathbf{z}) = G\mathbf{z}$, where the $r$'th coordinate of $\mathbf{z}$ is given by $z_r$ as in \eqref{eq:zr}. In particular, for $i \in [\ell]$, the $i$'th output of $\Rec(\mathbf{z})$ is
    \begin{align}
    \label{eq:recis}
    (G\mathbf{z})_i &= \sum_{r \in [n]} G[i,r] \cdot  z_r = \sum_{r \in [n]} G[i,r] \left( \sum_{\chi \in \cM_{\cL(r)}} \mathbf{e}_{(r, \chi)} \chi( \mathbf{y}_{\cL(r)} ) \right).
    \end{align}
    On the other hand, consider the quantity $f_i(\mathbf{x}^{(i)})$ that we want to recover, from \eqref{eq:want2}.  This is:
    \begin{align*}
    f_i(\mathbf{x}^{(i)}) &= \ \sum_{\fm \in \mathcal{M}} \psi_i(\fm) \cdot \fm(\mathbf{y}) \\
    &= \sum_{\fm \in \mathcal{M}} (S \cdot \mathbf{e})_{(i,\fm)} \cdot \fm(\mathbf{y}) \qquad \text{ as $S\cdot\mathbf{e} = \mathbf{g}$} \\
    &= \sum_{\fm \in \mathcal{M}} \left( 
        \sum_{r \in [n]} \sum_{\chi \in \cM_{\cL(r)}} S[ (i,\fm), (r, \chi) ] \cdot \mathbf{e}_{(r,\chi)}
    \right) \cdot \fm(\mathbf{y}) \\
      &= \sum_{\fm \in \mathcal{M}} \left( 
        \sum_{r \in [n]} \sum_{\chi \in \cM_{\cL(r)}} \mathbf{1}[\fm = \chi] G[i,r] \cdot \mathbf{e}_{(r,\chi)}
    \right) \cdot \fm(\mathbf{y}) \\
          &= \sum_{\fm \in \mathcal{M}} \left( 
        \sum_{r \in [n]} \mathbf{1}[\fm \in \mathcal{M}_{\cL(r)}] G[i,r] \cdot \mathbf{e}_{(r,\fm)}
    \right) \cdot \fm(\mathbf{y}) \\
    &= \sum_{r \in [n]} G[i,r] \left( \sum_{\chi \in \cM_{\cL(r)}} \mathbf{e}_{(r, \chi)} \chi( \mathbf{y}_{\cL(r)} ) \right) \\
    &= (G\mathbf{z})_i \qquad \text{ by \eqref{eq:recis} } \\
    &= \Rec(\mathbf{z})_i.
    \end{align*}
    Thus, $\Rec$ returns what it is supposed to with probability $1$, and the HSS scheme is correct.
\end{proof}

Given Claim~\ref{claim:correct}, we now only need to show that we can find a vector $\mathbf{e}$ so that $S \cdot \mathbf{e} = \mathbf{g}$.  To do this, we will show that the matrix $S$ has full row rank, and in particular we can solve the above affine system.

\begin{claim}\label{claim:fullrank}

Let $S$ be as above.  Then $S$ has full row rank.
\end{claim}
\begin{proof}
    For $\fm \in \cM$, consider the block $S^{(\fm)}$ of $S$ restricted to the rows $\{ (i, \fm) : i \in [\ell] \}$ and the columns $\{ (r, \fm) : r \in [n] \}$.  Thus, $S^{(\fm)} \in \F^{\ell \times n}$, and we may index the rows of $S^{(\fm)}$ by $i \in [\ell]$ and the columns by $r \in [n]$.  By the definition of $S$, we have
    \[S^{(\fm)}[i,r] = G[i,r] \cdot \mathbf{1}[\fm \in \cM_{\cL(r)}]. \]
    For $\Lambda \subseteq [s]$, let $G(\Lambda)$ denote the restriction of $G$ to the columns $r \in [n]$ so that $\cL(r) \in \Lambda$. For $\fm \in \cM$, let $\Lambda_\fm = \inset{ \lambda \in [s] : \fm \in \cM_\lambda }.$   Then 
    \begin{equation}\label{eq:samerank} \mathrm{rank}(S^{(\fm)}) = \mathrm{rank}(G(\Lambda_{\fm})),
    \end{equation}
    as we can make $G(\Lambda_{\fm})$ out of $S^{(\fm)}$ by dropping all-zero columns.
    We claim that the rank in \eqref{eq:samerank} is exactly $\ell$.  To see this, first
    note that $|\Lambda_{\fm}| \geq s - dt$.  Indeed, if 
    $ \fm = \prod_{k=1}^d y^{(i)}_{k, T_k} $ for some $\mathbf{T} \in \mathcal{T}^d$, then the only $\lambda \in [s]$ so that $\fm \not\in \cM_\lambda$ are those so that $\lambda \in \bigcup_{k=1}^d T_k$, and there are at most $dt$ such $\lambda$.  Let $\Lambda_{\fm}' \subseteq \Lambda_{\fm}$ be an arbitrary subset of size $s - dt$.  Then Corollary~\ref{cor: generator of good labelweight code is MDS like} implies that $\det(G(\Lambda_\fm')) \neq 0$, which in particular implies that $G(\Lambda_\fm)$ has full row rank.
    Equation~\ref{eq:samerank} then implies that $S^{(\fm)}$ also has full row rank, namely that it has rank $\ell$.

    Next, observe that, after rearranging rows and columns appropriately, $S$ is a block-diagonal matrix with blocks $S^{(\fm)}$ on the diagonal, for each $\fm \in \cM$.  Indeed, any entry $S[(i,\fm), (r, \chi)]$ where $\fm \neq \chi$ is zero by construction, and any entry where $\fm = \chi$ is included in the block $S^{(\fm)}$.  Thus, we conclude that $S$ too has full row-rank, as desired.
\end{proof}

Claim~\ref{claim:correct} says that any $\mathbf{e}$ so that $S\cdot\mathbf{e} = \mathbf{g}$ corresponds to a correct $\Eval$ function (with $\Share$ and $\Rec$ as given above), and Claim~\ref{claim:fullrank} implies that we can find such an $\mathbf{e}$ efficiently.  Thus we can efficiently find a description of $\pi = (\Share, \Eval, \Rec)$.  It remains to verify that $\textsf{DownloadRate}(\pi) \geq (s - dt)/s$, which follows from construction, as the download rate of $\pi$ is equal to the rate of $\cC$, which is by definition $(s-dt)/s$.\end{proof}

\section{Linear Codes with Good Labelweight}\label{sec:codes}

Now that we know that download-optimal linear HSS schemes are equivalent to linear codes with good labelweight, we focus on constructions and limitations of such codes.  In this section we give a construction, and a nearly-matching infeasibility result.  We begin by stating the constructive result, which as noted in the Introduction is a very slight improvement over the construction in \cite{FIKW22}.

\begin{theorem}[Construction of linear codes with good labelweight]\label{thm:lblwt}
    Let $\F$ be a finite field of size $q$.
    Let $s,d,t \in \mathbb{Z}^+$ such that $s-dt>0$. For all integers 
    \begin{equation*}
        j \geq \begin{cases}
            \log_q(s-1) & q^j \text{ odd, or } (s-dt) \not\in \lbrace 3, q^j-1 \rbrace\\
            \log_q(s-2) & q^j \text{ even, and } (s-dt) \in \lbrace 3, q^j-1 \rbrace
        \end{cases},
    \end{equation*}
    there is an explicit construction of a linear code $\mathcal{C} \subset \F^{js}$ of dimension $j(s-dt)$, and a labeling $\mathcal{L}:[js] \to [s]$ such that $\Delta_\mathcal{L}(\mathcal{C}) \geq dt+1$.

\end{theorem}

The following theorem says that Theorem~\ref{thm:lblwt} is basically optimal, in that we cannot take $j$ to be substantially smaller:

\begin{theorem}[Limitations on linear codes with good labelweight]\label{thm:lblwt_limitations}
    Let $\F$ be a finite field of size $q$.  Let $s,d,t \in \mathbb{Z}^+$ such that $s - dt > 0$. 
    Suppose that there is a code $\cC \subseteq \F^n$ with dimension $\ell$ and rate $(s-dt)/s$; and a labeling function $\cL:[n] \to [s]$ so that $\Delta_\cL(\cC) \geq dt + 1$.  Suppose that $j$ is such that $\ell = j(s- dt)$ and $n = js$.    Then $j$ is an integer, and 
  \[j \geq \lceil \max\inset{ \log_q(s -dt +1), \log_q(dt +1)} \rceil.\]
\end{theorem}
\begin{remark}[Gap between upper and lower bounds on $j$.]\label{rmk:gap}
As discussed in the Introduction,  Theorem~\ref{thm:lblwt_limitations} and Theorem~\ref{thm:lblwt} do not quite match.  However, because $j$ must be an integer, in fact the bounds are exactly the same for many parameter settings, especially when $q$ is large.  

We also remark that, when Theorems~\ref{thm:lblwt_limitations} and \ref{thm:lblwt} disagree, we conjecture that the \emph{construction} (Theorem~\ref{thm:lblwt}) is the correct one.  Indeed, our construction in Theorem~\ref{thm:lblwt} follows from ``puffing up'' a totally nonsingular (TN) matrix; the value of the parameter $j$ has to do with the field size over which this TN matrix is defined.  We show in the Section \ref{sec:MDS} that, assuming the MDS conjecture for extension fields, the field size that we use in our construction---and hence the value of $j$---is the best possible.  Thus, if Theorem~\ref{thm:lblwt} is not optimal, either the MDS conjecture is false over extension fields, or else there is an alternate way to construct Block TN matrices without going through TN matrices over larger fields.
\end{remark}

Given Lemma~\ref{lem:HSS_implies_lblwt} and Theorem~\ref{thm:lblwtImpliesHSS}, Theorems~\ref{thm:lblwt} and \ref{thm:lblwt_limitations} immediately imply the following corollary about linear HSS schemes.

\begin{corollary}[Classification of Amortization Parameter for Linear HSS Schemes]\label{cor:amortization}
    Let $s,d,t,m,\ell$ be positive integers, so that $m \geq d$.
    
    Let $\pi = (\Share, \Eval, \Rec)$ be a linear HSS scheme for some nontrivial (Definition~\ref{def:nontriv}) function class $\mathcal{F} \subseteq \POLYdmF$, with download rate $(s-dt)/s$.  Then $\ell = j(s - dt)$ for some integer $j$ that satisfies
    \[ j \geq \max\{ \lceil \log_q(s-dt + 1) \rceil, \lceil \log_q(dt + 1) \rceil \}. \]

    Conversely, let $j$ be any integer so that $j$ is as in Theorem~\ref{thm:lblwt}.  Then there are explicit constructions of linear HSS schemes for any $\mathcal{F} \subseteq \POLYdmF$ with download rate $(s-dt)/s$ and amortization parameter $\ell = j(s-dt)$.
\end{corollary}

\begin{remark}
    As in Remark~\ref{rmk:gap}, we observe that the bounds are quite close.  Indeed, for any setting of parameters, there is at most one value of $\ell$ (namely, $\ell = j(s-dt)$ for one particular integer $j$) for which we do not know whether or not there exists a download-optimal linear HSS scheme for $\POLYdmF^\ell$.  Moreover, for many parameter settings (especially when $q$ is large), the bounds exactly match.
\end{remark}

We prove Theorems~\ref{thm:lblwt} and \ref{thm:lblwt_limitations} later in this section.
The technical meat of the proofs is Lemma~\ref{lem:BTN_iff_goodLabelWt} (stated later), which characterizes codes of high labelweight in terms of \emph{block totally nonsingular matrices}, which we informally defined in the Introduction and which we formally define in Definition~\ref{def:blockMDS} below.   Theorems~\ref{thm:lblwt} and \ref{thm:lblwt_limitations} then follow from an analysis of block totally nonsingular matrices.

\subsection{Block Totally-Nonsingular Matrices}

Now that we understand the structure of optimal labeling functions, we move on to the connection to block totally-nonsingular matrices, which is the main technical result in this section.  We begin with some definitions.

We say that a matrix $A$ is \emph{totally nonsingular} if any square sub-matrix of $A$ (not necessarily contiguous) is nonsingular~\cite{gantmacher1980theory}. We extend the definition of total nonsingularity to block matrices as follows.

\begin{definition}[Block Totally-Nonsingular Matrices]\label{def:blockMDS}
Let $\F$ be a finite field, and fix integers $j,r,u$.  
    Let $\mathbf{A} = \left[ A_{i,k} \; : \; i \in [r], k \in [u] \right] \in \left(\text{GL}(\mathbb{F}, j) \right)^{r \times u}$ be a $r \times u$ block array of invertible $j \times j$ matrices $A_{i,k}$. We say that $\mathbf{A}$ is \emph{Block Totally Nonsingular} (Block TN) if every square sub-array of block matrices is full-rank. More precisely, for all $S \subseteq [r]$ and $S' \subseteq [u]$ of the same size, the matrix $\mathbf{A}' = [A_{i,i'} \in \mathbf{A} \; : \; i \in S, i' \in S' ] \in \F^{ j|S| \times j|S'|}$ is nonsingular.
\end{definition}

With Definition~\ref{def:blockMDS} out of the way, we can state and prove the main technical lemma of this section, which says that constructing codes with good labelweight is equivalent to constructing Block TN matrices:

\begin{lemma}\label{lem:BTN_iff_goodLabelWt}

    Let $j, s, d, t$ be positive integers so that $s > dt$, and let $q$ be a prime power.  Then the following are equivalent.
    \begin{itemize}
        \item[(i)] There is a linear code $\cC \subseteq \F_q^{js}$ with block length $js$ and dimension $j(s-dt)$ and a labeling $\cL: [js] \to [s]$ so that $\Delta_{\cL}(\cC) \geq dt + 1$.
        \item[(ii)] There exists a Block TN matrix $\mathbf{A} \in GL(\F_q,j)^{(s-dt) \times dt}$.
    \end{itemize}

    Moreover, the equivalence is constructive: if $\mathbf{A}$ is a Block TN matrix as in (ii), then the code $\cC$ in (i) is generated by the matrix $[I | \mathbf{A}]$.  If $\cC$ is a code as in (i), then it has a generator matrix of the form $[I | \mathbf{A}]$, where $\mathbf{A}$ is the Block TN matrix as in (ii). 
\end{lemma}

\begin{proof}

    We begin by showing that (i) implies (ii).
    Let $\cC$ be a code as in (i), and 
    let $G \in \mathbb{F}_q^{j(s-dt) \times js}$ be a generator matrix of $\mathcal{C}$. Corollary~\ref{cor: j is an int, labeling is wlog} implies that the labeling function $\mathcal{L}$ may be given by $\mathcal{L}:x \mapsto \lceil x/j \rceil$ after a re-ordering of the codeword coordinates; and Corollary~\ref{cor: generator of good labelweight code is MDS like} implies that after this re-ordering, the first $s-dt$ columns of any generator matrix $G$ for $\cC$ form a square invertible matrix, so we may put $G$ in systematic form by a sequence of row operations. Hence without loss of generality, we may write 
    \[ G = [ I | \mathbf{A} ]\]
    for some matrix $\mathbf{A} \in \F^{j(s-dt) \times jdt}$.
    We will view $G$ as a matrix of $j \times j$ blocks, as follows:
    \begin{equation*}
        G = \left[
        \begin{array}{c|c}
                    \begin{matrix}
                         I^{j \times j} &                  &            &         \\
                                        &  I^{j \times j}  &           &                     \\
                                        &                  & \ddots    &                       \\
                                        &                  &           & I^{j \times j}      
                    \end{matrix} & \left[ \begin{matrix}
                           A_{1,1} & A_{1,2} & \cdots & A_{1,dt}\\
                              A_{2,1} & A_{2,2} & \cdots & A_{2,dt}\\
                              \vdots  & \vdots  & \ddots & \vdots\\
                              A_{s-dt,1} & A_{s-dt,2} & \cdots & A_{s-dt, dt}
                    \end{matrix} \right]
        \end{array}
        \right] =: \left[\begin{array}{l}
             G_1 \\
             G_2 \\
                \vdots\\
                G_{s-dt}
        \end{array}\right].
    \end{equation*}
    where $G_1, G_2, \ldots, G_{s-dt} \in \mathbb{F}_q^{j \times js}$.

    We will show that $\mathbf{A}$ is a block TN matrix as in (ii).  The dimensions follow by construction, so we must show that every square sub-array of $\mathbf{A}$ is nonsingular.
    Let $\mathbf{A}' \in (GL(\F_q, j))^{u \times u}$ denote an arbitrary $u \times u$ block sub-array of $\mathbf{A}$; suppose that the $u$ block-rows in $\mathbf{A}'$ are indexed by $U \subseteq [s - dt]$ and the $u$ block-columns of $\mathbf{A}'$ are indexed by $V \subseteq [dt]$.
    
    Let $\mathbf{A}'' \in (GL(\F_q, j))^{u \times dt - u}$ denote the block sub-array that contains block-rows indexed by $U$ and block-columns indexed by $[dt] \setminus V$.  That is, $\mathbf{A}''$ is the ``rest'' of  $\mathbf{A}$, restricted to the rows that $\mathbf{A}'$ touches.

    To show that $\mathbf{A}'$ is full rank, we will show that for every vector $\mathbf{m} \in (\F^j)^u$, $\mathbf{m}^T \mathbf{A}' \neq \mathbf{0}$.  
    
    Suppose towards a contradiction that $\mathbf{m}^T \mathbf{A}' = \mathbf{0}$ for some $\mathbf{m}$.
    Consider the vector $\tilde{\mathbf{m}} \in (\F^j)^{(s - dt)}$ that is identified with $\mathbf{m}$ on $U$, and is $0$ elsewhere.
Then $\mathbf{c} := \tilde{\mathbf{m}}^T G$ is a codeword in $\cC$. Consider which blocks intersect the support of $\mathbf{c}$:
\begin{itemize}
    \item On the systematic coordinates (the first $s-dt$ blocks), the support of $\mathbf{c}$ is exactly the blocks indexed by $U$; there are $u$ of these.
    \item On the blocks indexed by $V$, the support of $\mathbf{c}$ is the support of $\mathbf{m}^T \mathbf{A}'$.  As we are assuming that this is zero, $\mathbf{c}$ has no support on these blocks.
    \item On the blocks indexed by $[dt] \setminus V$, there is at most $|[dt] \setminus V| = dt - u$ blocks of support.
\end{itemize}
But together, this implies that the number of blocks that the support of $\mathbf{c}$ touches is at most $u + (dt - u) = dt < dt + 1$.  Since by construction, each block corresponds to a set $\cL^{-1}(\lambda)$ for some $\lambda \in [s]$, this implies that $\Delta_\cL(\mathbf{c}) \leq dt$, which contradicts the assumption that $\Delta_{\cL}(\cC) \geq dt + 1.$  Thus $\mathbf{A}'$ is non-singular.  Since $\mathbf{A}'$ was arbitrary, we conclude that all such matrices are non-singular, and hence $\mathbf{A}$ is block TN, proving (ii).

    Now we show that (ii) implies (i).  Let $\mathbf{A}$ be as in (ii), and consider the matrix
    \[ G = [ I | \mathbf{A} ].\]
    Define $\cL: [js] \to [s]$ by $\cL(x) = \lceil x/j \rceil$, and let $\cC \subseteq \F_q^{js}$ be the code of dimension $j(s - dt)$ whose generator matrix is $G$.  Then $\cC$ has the block length and dimension in (i), so it remains to show that $\Delta_\cL(\cC) \geq dt + 1$.  Let $\mathbf{c} = \mathbf{m}^T G$ be an arbitrary codeword of $\cC$, where we view $\mathbf{m}$ as an element of $(\F_q^j)^{s -dt}$.  Let $U \subseteq [s - dt]$ be the support of $\mathbf{m}$ (that is, the set of \emph{blocks} on which it is nonzero); say that $|U| = u$.  Consider $\mathbf{m}^T \mathbf{A}$, viewed as an element of $(\F_q^j)^{dt}$.   
 Let $V \subseteq [dt]$ be the support of $\mathbf{m}^T \mathbf{A}$ (again, note that $V$ represents a set of \emph{blocks}).  We claim that $|V| \geq dt + 1  - u$.  Indeed, otherwise the sub-array $\mathbf{A}'$ of $\mathbf{A}$ whose block-rows are indexed by $U$ and whose block-columns are indexed
by $V^c$ would be a $u \times u$ sub-array so that $\mathbf{m}_U^T \mathbf{A}' = \mathbf{0}$ (where $\mathbf{m}_U$ represents the restriction of $\mathbf{m}$ to $U$).  Thus, $\mathbf{A}'$ would be singular, but this would contradict the block total nonsingularity of $\mathbf{A}$.  We conclude that $|V| \geq dt + 1 - u$, as claimed. 

Thus, the support of $\mathbf{c}$ includes $u$ blocks from the systematic coordinates, and at least $dt + 1 - u$ blocks from the non-systematic coordinates, for a total of at least $dt + 1$.  With our definition of $\cL$, this implies that $\Delta_\cL(\mathbf{c}) \geq dt + 1$.  Since $\mathbf{c}$ was arbitrary, this implies that $\Delta_\cL(\cC) \geq dt + 1$ as well.  This proves (i).
 
\end{proof}

\subsection{Understanding Block-TN Matrices: Proofs of Theorems~\ref{thm:lblwt} and \ref{thm:lblwt_limitations}}

Lemma~\ref{lem:BTN_iff_goodLabelWt} tells us that constructing linear codes with optimal labelweight is  equivalent to constructing block-TN matrices.  Thus, in this section, we investigate how well we can do this, which will lead to the Theorems~\ref{thm:lblwt} and \ref{thm:lblwt_limitations} stated earlier.

\subsubsection{Limitations on block-TN matrices}
We begin with infeasibility results.
In this section, we will prove the following lemma.
\begin{lemma}\label{lem:BNS_lowerbound}
    Let $r,u \in \mathbb{Z}^+$ such that $\min\lbrace r, u \rbrace > 1$.  Let $q$ be a prime power. Suppose that there exists a block-TN matrix $\mathbf{A} \in (GL(\F_q, j))^{r \times u}$.  Then
    \[ j \geq \lceil \max\inset{ \log_q(r+1), \log_q(u+1)} \rceil. \]
\end{lemma}
We note that, along with Lemma~\ref{lem:BTN_iff_goodLabelWt}, Lemma~\ref{lem:BNS_lowerbound} immediately implies Theorem~\ref{thm:lblwt_limitations}.

Before we prove Lemma~\ref{lem:BNS_lowerbound}, we make an observation about the maximum size of sets of invertible matrices whose differences are also invertible:
\begin{lemma}\label{lem: max size of additive GL group}
Let $\F$ be a finite field of size $q$, and 
    let $\mathcal{W} \subseteq \text{GL}(\mathbb{F}, j)$ be a set so that for all distinct $W, W' \in \mathcal{W}$, $\det(W - W') \neq 0 $. Then $|W| \leq q^j - 1$.
\end{lemma}
\begin{proof}
    First, we observe that no element $W \in \mathcal{W}$ can have a first row that is all zero, as each must be non-singular.  Similarly, any two elements $W, W' \in \mathcal{W}$ cannot have the same first row, otherwise $W-W'$ would have a zero first row and be singular. The bound on $|W|$ then follows immediately from the observation that there are only $q^j-1$ non-zero elements of $\mathbb{F}_q^j$. 
\end{proof}
We remark that the bound above is tight.  Indeed, consider the following example, which we will also use in our construction later.  (See, e.g., \cite{Wardlaw1994MatrixRO} for more details). 

\begin{example}\label{ex:embedding}
Fix a prime power $q$ and any integer $e$.
Let $\mathbb{E}$ be a finite field of size $q^e$, and let $\mathbb{F} \leq \mathbb{E}$ be a subfield of size $q$.
Then $\mathbb{E}$ is a vector space over $\F$ of dimension $e$.  Let $\mathcal{B} = \{b_1, \ldots, b_e\}$ be any basis for $\mathbb{E}$ over $\F$.  For any $\alpha \in \mathbb{E}$, multiplication by $\alpha$ is an $\F$-linear map, and in particular it can be represented in the basis $\mathcal{B}$ by a matrix $M_\alpha \in \F^{e \times e}$.  Let  $\varphi:\mathbb{E} \to \F^{e \times e}$ be the embedding so that $\varphi: \alpha \mapsto M_\alpha$.

Then let $\mathcal{W} = \varphi(\mathbb{E}^*)$.  
We have $|\mathcal{W}| = q^e - 1$.  Further, every element of $\mathcal{W}$ is invertible, as $\varphi(\alpha)^{-1} = \varphi(\alpha^{-1})$.  Further, since $\varphi(\alpha) - \varphi(\beta) = \vphi(\alpha - \beta)$, $\mathcal{W}$ is closed under differences and in particular $W - W'$ is also invertible for any distinct $W,W' \in \mathcal{W}$.  
\end{example}

Now we can prove Lemma~\ref{lem:BNS_lowerbound}.
\begin{proof}[Proof of Lemma~\ref{lem:BNS_lowerbound}]
    Let $\mathbf{A}$ be a block-TN matrix with entries $A_{i,k} \in \text{GL}(\mathbb{F}_q,j)$:
    \begin{equation*}
        \mathbf{A} = \left[ \begin{matrix}
                           A_{1,1} & A_{1,2} & \cdots & A_{1,u}\\
                              A_{2,1} & A_{2,2} & \cdots & A_{2,u}\\
                              \vdots  & \vdots  & \ddots & \vdots\\
                              A_{r,1} & A_{r,2} & \cdots & A_{r, u}
                    \end{matrix} \right].
    \end{equation*}
   Notice that the property of being Block-TN is maintained under elementary row operations within a given block-row $\mathbf{A}[i,:]$ or elementary column operations within a given block-column $\mathbf{A}[:,k]$. Hence
    \begin{equation*}
        \mathbf{A}' = \left[ \begin{array}{cccc}
                           I & A_{1,1}^{-1}A_{1,2} & \cdots &A_{1,1}^{-1} A_{1,u}\\
                              I & A_{2,1}^{-1}A_{2,2} & \cdots & A_{2,1}^{-1}A_{2,u}\\
                              \vdots  & \vdots  & \ddots & \vdots\\
                              I & A_{r,1}^{-1} A_{r,2} & \cdots & A_{r,1}^{-1} A_{r, u}
                    \end{array} \right]
    \end{equation*}
    is also block TN.  In particular, for any distinct $\rho, \sigma \in [r]$ and $\mu \in [u] \setminus \lbrace 1 \rbrace$, we must have
    \begin{equation*}
        0 \neq \det\left( \left[ \begin{matrix}
            I & A_{\rho, 1}^{-1} A_{\rho, u}\\
            I & A_{\sigma, 1}^{-1} A_{\sigma, u}
        \end{matrix} \right] \right) = \det\left(  A_{\sigma, 1}^{-1} A_{\sigma, u} - A_{\rho, 1}^{-1} A_{\rho, u} \right),
    \end{equation*}
    where the non-equality is due to the assumption of block total nonsingularity, and the equality is due to the fact that
    \[ \det\begin{bmatrix} I & W \\ I & W' \end{bmatrix} = \det(W' - W)\]
    for any matrices $W,W'$.
    It follows that for all $\mu \in [u]$, $\mathbf{A}'[:, \mu]$ is a $r$-tuple of distinct elements of $\text{GL}(\mathbb{F}_q,j)$, satisfying the condition that the sum of any two coordinates is again an element of $\text{GL}(\mathbb{F}_q,j)$. Then by Lemma \ref{lem: max size of additive GL group}, we have $r \leq q^{j}-1$ so that $j \geq \log_{q}(r+1)$. By an identical argument, we also have that $j \geq \log_q(u+1)$, and this yields the result.
\end{proof}

\subsubsection{Connection to MDS Conjecture}\label{sec:MDS}

Before we give our construction (proving Theorem~\ref{thm:lblwt}), we take a quick detour to expain the connection between such constructions and the \emph{MDS conjecture} (Conjecture~\ref{conj: MDS conjecture} below) over extension fields, which is a long-standing open problem.  This connection will inform our construction in Theorem~\ref{thm:lblwt}, and also suggest that we may not be able to do better.

As we will see formally below in Lemma~\ref{lem:TNtoBlockTN}, one way to construct Block TN matrices with entries in $\text{GL}(\mathbb{F}_q, j)$ is to construct (non--block) TN matrices with entries in $\mathbb{F}_{q^j}$, and then ``puff up'' each entry using the embedding $\vphi$ from Example~\ref{ex:embedding}.  We show below in Theorem~\ref{thm: mds iff tns} that constructing TN matrices is equivalent to constructing \emph{MDS matrices} (that is, matrices that are generator matrices of MDS codes) over extension fields, which are well-studied.

As a corollary of that result, we get the main result for this section: \emph{assuming} the MDS conjecture, we can characterize the best possible TN matrices we can get, and thus the best possible block-TN matrices that we can get via the method ``start with a TN matrix and apply $\vphi$.''  This is formally stated as the following corollary.

\begin{corollary}\label{cor: mds conj applied to hss parameters}
    Assume the MDS conjecture. Then there exists an (efficiently computable) totally nonsingular matrix $A \in \left(\mathbb{F}_{q^j} \setminus \lbrace 0 \rbrace\right)^{(s-dt) \times s}$ if and only if $s-dt \leq q^j$ and
    \begin{equation*}
        j \geq \begin{cases}
            \log_q(s-1) & q^j \text{ odd, or } (s-dt) \not\in \lbrace 3, q^j-1 \rbrace\\
            \log_q(s-2) & q^j \text{ even, and } (s-dt) \in \lbrace 3, q^j-1 \rbrace
        \end{cases}.
    \end{equation*}
    Furthermore, the assumption of the MDS conjecture is necessary only for the forward direction.
\end{corollary}

Before we state Theorem~\ref{thm: mds iff tns} and prove it and Corollary~\ref{cor: mds conj applied to hss parameters}, we briefly recall the MDS conjecture.  First, we formally define MDS codes/matrices.

\begin{definition}\label{MDS matrix}
Let $\cC \subseteq \F^u$ be a linear code of dimension $r$ and block length $u$.  We say that $\cC$ is \emph{Maximum Distance Separable} (MDS) if the Hamming distance of $\cC$ is $u - r + 1$; equivalently, any $r$ symbols of a codeword $\mathbf{c} \in \cC$ determine the rest of $\mathbf{c}$.

We say that a matrix $G \in \F^{r \times u}$ is an \emph{MDS matrix} if it is the generator matrix for an MDS code.  Equivalently, $G$ is an MDS matrix if and only if any $r \times r$ submatrix of $G$ has full rank.
\end{definition}

The MDS conjecture concerns the minimum possible field size over which MDS matrices can exist.  It is known that a slight variant on \emph{Reed-Solomon} codes can yield a construction over reasonably small fields.  For completeness, we provide the construction that proves Theorem~\ref{thm: MDS matrices meeting mds conjecture} in Appendix~\ref{sec: rs with extra col (apx)}.
\begin{theorem}[See, e.g., \cite{Wolf1969},\cite{Ball2012OnSO}]\label{thm: MDS matrices meeting mds conjecture}
    Let $\mathbb{F}$ be a field of order $q^j$. For any $r, u \in \mathbb{Z}^+$ satisfying $r \leq q^j$ and
    \begin{equation*}
        u \leq \begin{cases}
            q^j +1 & q^j \text{ odd, or } r \not\in \lbrace 3, q^j-1 \rbrace\\
            q^j +2 & q^j \text{ even, and } r \in \lbrace 3, q^j-1 \rbrace
        \end{cases}
    \end{equation*}
    there exists an explicit MDS matrix $A \in \mathbb{F}_{q^j}^{r \times u}$, $r\leq u$.
\end{theorem}

The \emph{MDS conjecture} says that the above bounds are optimal.
\begin{conjecture}[MDS Conjecture \cite{Segre1955CurveRN}] \label{conj: MDS conjecture}
    The bounds of Theorem \ref{thm: MDS matrices meeting mds conjecture} are tight.
\end{conjecture}
The MDS conjecture has been proven in prime order fields, and it is known to hold in subcases of extension fields \cite{Ball2012OnSO}.

Now we can state and prove Theorem~\ref{thm: mds iff tns}, which says that an MDS matrix is equivalent to a TN matrix.

\begin{theorem}\label{thm: mds iff tns}
    Let $r,u \in \mathbb{Z}^+$ with $r \leq u$. Then there exists a totally nonsingular matrix $A \in \left(\mathbb{F}_{q^j} \setminus \lbrace 0 \rbrace\right)^{r \times u}$ if and only if there exists a  MDS matrix $A' \in \mathbb{F}_{q^j}^{r \times (r+u)}$. 
    Furthermore, this biconditional is constructive:
    \begin{itemize}
        \item[(i)] if $A \in \left(\mathbb{F}_{q^j} \setminus \lbrace 0 \rbrace\right)^{r \times u}$ is a totally nonsingular matrix, then $[ I^{r \times r} \mid  A]$ is a $r \times (r+u)$ MDS matrix; and

        \item[(ii)] if $A' \in \mathbb{F}_{q^j}^{r \times (r+u)}$ is an MDS matrix, then $A' = [ I^{r \times r} \mid  A ]$ without loss of generality, where $A$ is a totally nonsingular matrix.
    \end{itemize}
\end{theorem}

To prove Theorem \ref{thm: mds iff tns}, we prove each direction separately in Lemmas \ref{lem: total non-singular implies a larger MDS matrix (apx)} and \ref{lem: mds matrix implies total nonsingular (apx)}, respectively. We then provide the proof of Corollary \ref{cor: mds conj applied to hss parameters}.

\begin{lemma}\label{lem: total non-singular implies a larger MDS matrix (apx)}
    Let $A \in \left(\mathbb{F}_{q^j} \setminus \lbrace 0 \rbrace\right)^{r \times u}$ be a totally non-singular matrix with $r \leq u$. Then there exists a $r \times (r + u)$ MDS matrix $A' \in \mathbb{F}_{q^j}^{r \times (r+u)}$.
\end{lemma}
\begin{proof}
Let $A' = [A | I^{r \times r}] \in \F_{q^j}^{r \times (r+u)}$. We claim that $A'$ is an MDS matrix. Since any $r \times r$ submatrix of $A$ is nonsingular, and since $I^{r\times r}$ is also non-singular, we need only consider $r \times r$ submatrices $B$ of $A'$ that contain some $\rho \in \{0, \ldots, r-1\}$ of the last $r$ columns of $A'$. Let $\hat{A}$ be any $r \times (r-\rho)$ submatrix of $A$,  partitioned into an upper $(r- \rho) \times (r-\rho)$ square matrix and a lower $\rho \times (r-\rho)$ matrix as follows:
    \begin{equation*}
        \hat{A} = \left[
        \begin{array}{cc}
            \\
               \hat{A}_{\text{up}}\in \mathbb{F}_{q^j}^{(r-\rho) \times (r-\rho)}\\ 
               \\
               \hat{A}_{\text{down}}\in \mathbb{F}_{q^j}^{\rho \times (r-\rho)} \\
               \\
            \end{array}
        \right]
    \end{equation*}
    Then up to a permutation of rows, we may write any $r \times r$ submatrix $B$ of $A'$ containing $\rho$ of the last $r$ columns as
    \begin{equation*}
        B:= \left[
            \begin{array}{cc}
            \\
               \hat{A}_{\text{up}}  & 0^{(r-\rho) \times \rho} \\ 
               \\
               \hat{A}_{\text{down}} & I^{\rho \times \rho}\\
               \\
            \end{array}
        \right].
    \end{equation*}
    However, such a $B$ is non-singular, as $\det(B)=\det(\hat{A}_{\text{up}}) \neq 0$.  Thus, $A'$ is MDS.
\end{proof}

Next, we prove the other direction of Theorem~\ref{thm: mds iff tns}.
\begin{lemma}\label{lem: mds matrix implies total nonsingular (apx)}
     Let $A' \in \mathbb{F}_{q^j}^{r \times (r+u)}$ be an MDS matrix with $r \leq u$. Then there exists a totally non-singular matrix $A \in \left(\mathbb{F}_{q^j} \setminus \lbrace 0 \rbrace\right)^{r \times u}$.
\end{lemma}
\begin{proof}
    Without loss of generality, assume that $A' = [I \vert A]$ is in reduced row-echelon form, with $A \in \mathbb{F}_{q^j}^{r \times u}$. We will show that $A$ is totally non-singular. Let $R \subseteq [r], U \subseteq [u]$ be so that $|R| = |U| = \rho$, for some $0 \leq \rho \leq r$.  Thus, $R$ and $U$ represent row and column indices of some (not necessarily contiguous) square submatrix of $A$; denote this matrix by $V\in \mathbb{F}_{q^j}^{\rho \times \rho}$.  Thus, we want to show that $V$ is full rank.  
    
    Denote by $W\in \mathbb{F}_{q^j}^{\rho \times (u-\rho)}$ the submatrix of $A$ with rows indexed by $R$ and columns indexed by $[u] \setminus U$.
    Let $S'$ be a submatrix of consisting of $A'$ consisting of rows indexed by $U$.  Thus, up to a permutation of rows and columns, we may write $S'$ as
    \begin{equation*}
        S' = \left[ \begin{array}{cccc}
            I^{\rho \times \rho} & 0^{\rho \times (r- \rho)} & V & W
        \end{array} \right].
    \end{equation*}
    Now assume towards a contradiction that there exists some $\mathbf{m} \in \mathbb{F}_{q^j}^{\rho}$ such that $\mathbf{m}^T S' = \mathbf{0}$. Then, in particular, we have
    \begin{equation*}
        \mathbf{m}^T S'' := \mathbf{m}^T \left[ \begin{array}{cc}
            0^{\rho \times (r-\rho)} & V
        \end{array} \right] = \mathbf{0}.
    \end{equation*}
    But $S''$ is a $\rho \times r$ submatrix of $A'$, which is contained in an $r \times r$ submatrix of $A'$; if $S''$ is not full-rank, then there exists a singular $r \times r$ submatrix of $A'$, contradicting the premise that $A'$ is an MDS matrix.
\end{proof}

Thus, we have established both directions of Theorem~\ref{thm: mds iff tns}.  This proves the theorem.

We are now equipped to prove Corollary \ref{cor: mds conj applied to hss parameters}.

\begin{proof}[Proof of Corollary \ref{cor: mds conj applied to hss parameters}.]
    In the forward direction, note that the total non-singularity of $A$ implies $s-dt \leq q^j$. By Theorem \ref{thm: mds iff tns}, there exists a MDS matrix $A' \in \mathbb{F}_{q^j}^{(s-dt) \times s}$. Considering the pair of bounds $s \leq q^j + 1, q^j + 2$ of the MDS conjecture yields the desired result.

    In the reverse direction, Theorem \ref{thm: MDS matrices meeting mds conjecture} guarantees the existence of a MDS matrix $A' \in \mathbb{F}_{q^j}^{(s-dt) \times s}$ meeting the bounds of the MDS conjecture. We may assume $A' = \left[ I^{(s-dt) \times (s-dt)} \mid A \right]$ where $A \in \mathbb{F}_{q^j}^{(s-dt) \times dt}$; by Theorem \ref{thm: mds iff tns}, $A$ must be a totally non-singular matrix.
\end{proof}

\subsubsection{Near-optimal construction of block-TN matrices}

Next, we move onto constructions of block-TN matrices, which will establish Theorem~\ref{thm:lblwt}. We begin with Lemma~\ref{lem:TNtoBlockTN} below, which states that we can construct a block-TN matrix out of a TN matrix.

\begin{lemma}\label{lem:TNtoBlockTN}
Let $q$ be a prime power and let $e$ be any integer; let $\F \leq \mathbb{E}$ be finite fields so that $|\F| = q$ and $|\mathbb{E}| = q^e$.
Let $r,u$ be positive integers.

    Let $\mathbf{B} \in \left(\mathbb{E}\right)^{r \times u}$ be a totally nonsingular matrix. Let $\varphi: \mathbb{F} \to GL(\F, e)$ be as in Example~\ref{ex:embedding}.   Let $\mathbf{A} = \vphi(\mathbf{B}) \in \text{GL}(\mathbb{F},e)^{r \times u} $ be the $r \times u$ array of invertible matrices obtained by applying $\vphi$ element-wise to $\mathbf{B}$. Then $\mathbf{A}$ is block totally nonsingular.  
\end{lemma}

\begin{proof}
    We show the contrapositive.  Suppose that $\mathbf{A}$ is not block-TN. Then there exists some square subarray $\mathbf{A}'$ of $\mathbf{A}$ that is singular; say $\mathbf{A}' \in (GL(\F,e))^{a \times a}$. Since $\vphi$ is invertible, we may consider $\vphi^{-1}$ applied to each element of $\mathbf{A}$, which yields a square submatrix $\mathbf{B}' \in \F^{a \times a}$ of $\mathbf{B}$.  We claim that $\mathbf{B}'$ is singular.  Indeed, let $\mathbf{v} = \mathbf{v}^{(1)} \circ \cdots \circ \mathbf{v}^{(a)} \in (\F^e)^a$ be a kernel vector for $\mathbf{A}$.  Recall that $\vphi$ was formed by picking a basis $\mathcal{B}$ for $\mathbb{E}$ over $\F$.  For each $\mathbf{v}^{(i)} \in \F^e$, let $\gamma^{(i)} \in \mathbb{E}$ be the field element represented by $\mathbf{v}^{(i)}$ in the basis $\mathcal{B}$.  
    Then it follows from the definition of $\varphi$ that the vector $(\gamma^{(1)}, \ldots, \gamma^{(a)}) \in \mathbb{E}^a$ is a kernel vector for $\mathbf{B}'$, so $\mathbf{B}'$ is singular, as claimed. Thus, $\mathbf{B}$ is not totally nonsingular.
\end{proof}

Finally, we instantiate Lemma~\ref{lem:TNtoBlockTN} with the following Corollary of Theorem~\ref{thm: mds iff tns}:
\begin{corollary}[Existence of Block TN Matrices]\label{cor:Exist_blockTN_matrices}
    Let $\F$ be a finite field of size $q$.  Let $s,d,t \in \mathbb{Z}^+$ so that $s -dt > 0$.  For all integers
    \[ j \geq \begin{cases} \log_q(s-1) & q^j \text{ odd, or } (s - dt) \not\in \{3, q^j - 1\} \\ \log_q(s-2) & q^j \text{ even, and } (s-dt) \in \{3, q^j -1 \} \end{cases}\]
    there is an explicit construction of a block totally nonsingular matrix $A \in (GL(\F, j))^{(s-dt) \times dt}$. 
\end{corollary}
\begin{proof}
    By Theorem~\ref{thm: MDS matrices meeting mds conjecture} (with $r \gets s - dt$ and $u \gets s$), there is an MDS matrix $A' \in (\F_{q^j})^{(s-dt) \times s}$ as long as $j$ is as in the statement of the Corollary.  By Theorem~\ref{thm: mds iff tns}(ii), if we write $A' = [I | B]$, then $B$ is totally non-singular.  By Lemma~\ref{lem:TNtoBlockTN}, if we apply the embedding $\varphi: \F_{q^j} \to \F^{j \times j}$ from Example~\ref{ex:embedding} coordinatewise to $B$, we obtain a matrix $A \in (GL(\F,j))^{(s-dt) \times dt}$ that is block TN.
\end{proof}

Finally, we observe that Corollary~\ref{cor:Exist_blockTN_matrices}, along with Lemma~\ref{lem:BTN_iff_goodLabelWt}, implies Theorem~\ref{thm:lblwt}.  

This completes the proofs of both Theorems~\ref{thm:lblwt} and \ref{thm:lblwt_limitations}, providing nearly matching bounds on good labelweight codes, and finally establishing Corollary~\ref{cor:amortization}, which nearly pins down the allowable amortization parameters $\ell$ for download-optimal linear HSS schemes for polynomials.

\section{Acknowledgements}
We thank Yuval Ishai and Victor Kolobov for helpful conversations, and the anonymous referees for helpful feedback.

\bibliographystyle{alpha}
\bibliography{refs.bib}

\appendix
\section{Explicit constructions of MDS codes meeting the MDS conjecture}\label{sec: rs with extra col (apx)}
For completeness we include the explicit construction of MDS matrices which meet the bounds of the MDS conjecture with equality (see, e.g., \cite{Wolf1969, Ball2012OnSO}). Let $\mathbb{F}$ be a field of order $q^j$ and $r,u \in \mathbb{Z}^+$ satisfying $r \leq q^j$. We consider the following two cases.

\subsection{\texorpdfstring{$q^j$}{Lg} odd, or \texorpdfstring{$r \not\in \lbrace 3, q^j - 1 \rbrace$}{Lg}}

Let $\mathbb{F}_{q^j} = \lbrace \alpha_1, \alpha_2, \ldots, \alpha_{q^j}\rbrace$. Then
\begin{equation*}
    G := \left[
\begin{array}{ccccc}
     1 & 1 & \cdots & 1 & 0 \\
     \alpha_1 & \alpha_2 & \cdots & \alpha_{q^j} & 0\\
     \alpha_1^2 & \alpha_2^2 & \cdots & \alpha_{q^j}^2 & 0\\
     \vdots & \vdots & \ddots & \vdots & \vdots\\
     \alpha_1^{r-2} & \alpha_2^{r-2} & \cdots & \alpha_{q^j}^{r-2} & 0\\
     \alpha_1^{r-1} & \alpha_2^{r-1} & \cdots & \alpha_{q^j}^{r-1} & 1
\end{array}
    \right]
\end{equation*}
generates an MDS linear code of length $u = q^j +1$. Indeed, any $r \times r$ submatrix of $G$ comprised of any $r$ of its first $q^j$ columns is a non-singular Vandermonde matrix, so it remains only to observe that any $r \times r$ submatrix of $G$ including its final column has the form
\begin{equation*}
    G' := \left[
\begin{array}{ccccc}
     1 & 1 & \cdots & 1 & 0 \\
     \alpha_1 & \alpha_2 & \cdots & \alpha_{r-1} & 0\\
     \alpha_1^2 & \alpha_2^2 & \cdots & \alpha_{r-1}^2 & 0\\
     \vdots & \vdots & \ddots & \vdots & \vdots\\
     \alpha_1^{r-2} & \alpha_2^{r-2} & \cdots & \alpha_{r-1}^{r-2} & 0\\
     \alpha_1^{r-1} & \alpha_2^{r-1} & \cdots & \alpha_{r-1}^{r-1} & 1
\end{array}
    \right].
\end{equation*}
Note that the upper-left $(r-1) \times (r-1)$ submatrix of $G'$ is a non-singular square Vandermonde matrix; hence there does not exist any non-trivial element of the right kernel of $G'$.

\subsection{\texorpdfstring{$q^j = 2^h$}{Lg} even and \texorpdfstring{$r \in \lbrace 3, q^j -1 \rbrace$}{Lg}}

We consider only the case $r=3$. The case $r = q^j-1$ is similar; full details may be found in \cite{Wolf1969}. Let $\mathbb{F}_{2^h} = \lbrace \alpha_1, \alpha_2, \ldots, \alpha_{2^h}\rbrace$. Then
\begin{equation*}
    G = \left[
\begin{array}{cccccc}
     1 & 1 & \cdots & 1 & 0 & 0 \\
     \alpha_1 & \alpha_2 & \cdots & \alpha_{2^h} & 1 & 0 \\
     \alpha_1^2 & \alpha_2^2 & \cdots & \alpha_{2^h}^2 & 0 & 1
\end{array}
    \right]
\end{equation*}
generates a linear MDS code of length $u = 2^h + 2$; it suffices to observe the following determinants of any $3 \times 3$ submatrix that is not a Vandermonde matrix:
\begin{align*}
    \det\left( \left[ 
    \begin{array}{ccc}
        1 & 1 & 0 \\
        \alpha_1 & \alpha_2 & 0\\
        \alpha_1^2 & \alpha_2^2 & 1
    \end{array}\right] \right) &= \alpha_1 + \alpha_2 \neq 0\\
    \det\left( \left[ 
    \begin{array}{ccc}
        1 & 1 & 0 \\
        \alpha_1 & \alpha_2 & 1\\
        \alpha_1^2 & \alpha_2^2 & 0
    \end{array}\right] \right) &= \alpha_1^2 + \alpha_2^2 \neq 0\\
    \det\left( \left[ 
    \begin{array}{ccc}
        1 & 0 & 0 \\
        \alpha_1 & 1 & 0\\
        \alpha_1^2 & 0 & 1
    \end{array}\right] \right) &=  1.
\end{align*}

\section{Example Construction}\label{sec: example (apx)}
\begin{example}[Example of the Construction in Theorem~\ref{thm:lblwtImpliesHSS}]\label{example:construction}
    In this example, we consider the setting where
    $m =d = t = 1$ and $s=3$. 
    Note that since $d=1$, we are focusing on the problem of \emph{HSS for concatenation}; that is, the output client simply wants to recover all of the secrets.
    By Theorem~\ref{thm:FIKW rate LB}, the best possible download rate is $(s-dt)/s = 2/3$, and we will see how to achieve this.  
    
    We set $j=1$, $\ell = j(s-dt) = 2$, and $n = js = 3$.  Thus, we want to amortize over $\ell =2$ instances.  Suppose that the two secrets to be shared are $a,b \in \F_2$.\footnote{Note that in our standard notation, these would be called $x^{(1)}$ and $x^{(2)}$, but for this example with go with $a$ and $b$ to avoid excessive superscripts.}  

    We will construct $\pi = (\Share, \Eval, \Rec)$ as in the proof of Theorem~\ref{thm:lblwtImpliesHSS}.  In particular, $\Share$ will be $1$-CNF sharing among three servers, $s_1, s_2, s_3$.  That is, we choose $a_i, b_j$ uniformly at random so that
    \[ a = a_1 + a_2 + a_3 \qquad b = b_1 + b_2 + b_3,\]
    and then distribute the $a_i, b_j$ to the three servers as follows:
    
\begin{center}
    \begin{tabular}{|c|c|}
    \hline
    \textbf{Server} & \textbf{Local Shares} \\
    \hline
    \( s_1 \) & \( \{ a_2, a_3, b_2, b_3 \} \) \\
    \hline
    \( s_2 \) & \( \{ a_1, a_3, b_1, b_3 \} \) \\
    \hline
    \( s_3 \) & \( \{ a_1, a_2, b_1, b_2 \} \) \\
    \hline
    \end{tabular}
\end{center}

Next, we need to define $\Rec$, which as in the proof of Theorem~\ref{thm:lblwtImpliesHSS}, we will do using the generator matrix of a code with good labelweight.  
Let $\mathcal{L}:[3]\to [3]$ be defined by $\mathcal{L}: x \mapsto x$ and 
\begin{equation}\label{eq: example labelweight code}
    G = \left[
    \begin{matrix}
        1 & 0 & 1 \\
        0 & 1 & 1
    \end{matrix}
    \right].
\end{equation}
It is easily verified that $G$ describes a $[js = 3, j(s-dt)=2]_2$ linear code $\mathcal{C}$ such that $\Delta_\mathcal{L}(\mathcal{C}) \geq dt+1 = 2$. Recall from Corollary \ref{cor: generator of good labelweight code is MDS like} that for any set $\Lambda$ of $s-dt = 2$ servers that the submatrix of $G$ given by $G(\Lambda)$ is a $ j(s-dt) \times j(s-dt) = 2 \times 2$ invertible matrix. For instance, if $\Lambda = \lbrace 1,3 \rbrace$, we have
\begin{equation*}
    G(\Lambda) = \left[
    \begin{matrix}
        1 & 1  \\
        0 & 1 
    \end{matrix}
    \right].
\end{equation*}

The matrix $G$ defines $\Rec$ as follows.  Each of the three servers will send a single symbol: server $\cL(1)=1$ will send the first symbol, server $\cL(2) = 2$ will send the second symbol, and server $\cL(3) = 3$ will send the third.  Then we will group these symbols into a vector $\mathbf{z} \in \F_2^3$, and define 
\[ \Rec(\mathbf{z}) = G\mathbf{z}.\]
We want to define $\Eval$ so that this is equal to the vector $\begin{pmatrix} a \\ b \end{pmatrix}$.  In particular, we want to define $\Eval$ so that the sum of $s_1$'s and $s_3$'s output shares to be $a$, and the sum of $s_2$ and $s_3$'s output shares to be $b$.  In order to enforce this, we set up a linear system.

Let $\mathcal{M} = ( a_1, a_2, 
\ldots, b_3 )$ be the ordered set of monomial labels, and let $\mathcal{M}_i$ be the subset of $\mathcal{M}$ that is locally computable by $s_i$. It follows that the output shares $z_1, z_2, z_3$ downloaded from servers $s_1, s_2, s_3$, respectively, have the form
\begin{align}
    z_1 &= e_{(1, a_2)} \cdot a_2 + e_{(1, a_3)} \cdot a_3 + e_{(1, b_2)} \cdot b_2 + e_{(1, b_3)} \cdot b_3 \nonumber\\
    z_2 &= e_{(2, a_1)} \cdot a_1 + e_{(2, a_3)} \cdot a_3 + e_{(2, b_1)} \cdot b_1 + e_{(2, b_3)} \cdot b_3 \label{eq: example output shares} \\
    z_3 &= e_{(3, a_1)} \cdot a_1 + e_{(3, a_2)} \cdot a_2 + e_{(3, b_1)} \cdot b_1 + e_{(3, b_2)} \cdot b_2 \nonumber
\end{align}
where the coefficients $e_{(r,\chi)} \in \mathbb{F}_2$ are indexed by $(r,\chi) \in [n] \times \mathcal{M}_{\mathcal{L}(r)}$ determine whether a server indexed by $\mathcal{L}(r)$ includes in its output share $z_r$ the monomial $\chi \in \mathcal{M}_{\mathcal{L}(r)}$, so that 
\begin{equation*}
    z_1 + z_3 = a = a_1 + a_2 + a_3 \qquad \text{and} \qquad z_2 + z_3 = b = b_1 + b_2 + b_3.
\end{equation*}
Let $\mathbf{e}$ denote the vector of all such coefficients $e_{(r,\chi)}$.  Thus, $\mathbf{e}$ determines the function $\Eval$ with the interpretation above.

We can solve for the $\mathbb{e}$ by constructing the following linear system. Define a matrix $S \in \mathbb{F}_2^{\ell \abs{\mathcal{M}} \times \sum_{r \in [n]} \abs{\mathcal{M}_r}} = \mathbb{F}_2^{12 \times 12}$ as follows.  (This mirrors the matrix $S$ in the proof of Theorem~\ref{thm:lblwtImpliesHSS}).
\begin{itemize}
    \item[$\bullet$] The rows of $S$ are indexed by pairs $(i, \mathfrak{m}) \in [\ell] \times \mathcal{M}$. In this example, these are the pairs $(1,a_1), (1, a_2), \ldots, (2,b_3)$. 
    \item[$\bullet$] The columns of $S$ are indexed by pairs $(r, \chi)$ for $r \in [n]$ and $\chi \in \mathcal{M}_{\mathcal{L}(r)}$. In this example, for $r = 1$, these are the pairs $(1,a_2), (1,a_3), (1, b_2), (1, b_3)$.
\end{itemize}
Then the entry of $S$ indexed by $S\left[ (i, \mathfrak{m}), (r, \chi) \right]$ is given by
\begin{equation}\label{eq: construction of example system}
    S\left[ (i, \mathfrak{m}), (r, \chi) \right] = \begin{cases}
        1 & \mathfrak{m} = \chi \;\land\; G[i,r] \neq 0\\
        0 & \text{else}
    \end{cases}.
\end{equation}
We construct $S$ explicitly in Figure \ref{fig:example linear system}.

\begin{figure}
    \centering   

  \begin{equation*}
    S \mathbf{e} :=\begin{array}{l}
1, a_1 \\
1, a_2 \\
1, a_3 \\
1, b_1 \\
1, b_2 \\
1, b_3 \\ \hline
2, a_1 \\
2, a_2 \\
2, a_3 \\
2, b_1 \\
2, b_2 \\
2, b_3
\end{array}
\left[
\begin{array}{cccc|cccc|cccc}
 & & & & & &  & & 1 & & & \\
1 & & & & & & &  & &1 & & \\
 & 1 & & & & & & & & & & \\
 & & & & & & & & & & 1 & \\
 & &1 &  & & & & & & & &1 \\
 & & &1 &  & & & & & &  & \\ \hline
 & & & &1 & & & &1 & & & \\
 & & & & & & & & &1 & & \\
 & & & & &1 & & & & & & \\
 & & & & & &1 & & & &1 & \\
 & & & & & & & & & & & 1 \\
 & & & & & & &1 & & & & \\
\end{array}
\right]
\left[
\begin{array}{c}
e_{(1, a_2)} \\
e_{(1, a_3)} \\
e_{(1, b_2) }\\
e_{(1, b_3)} \\\hline
e_{(2, a_1)} \\
e_{(2, a_3) }\\
e_{(2, b_1)} \\
e_{(2, b_3)} \\\hline
e_{(3, a_1) }\\
e_{(3, a_2)} \\
e_{(3, b_1)} \\
e_{(3, b_2)} \\
\end{array}
\right]
=
\left[
\begin{array}{c}
1 \\
1 \\
1 \\
0 \\
0 \\
0 \\ \hline
0 \\
0 \\
0 \\
1 \\
1 \\
1 \\
\end{array}
\right] =: \mathbf{g}
\end{equation*}
 \caption{Linear system describing the contents of each server's output shares.}
    \label{fig:example linear system}
\end{figure}

There exists a solution to $S\mathbf{e} = \mathbf{g}$ as long as the columns of $S$ span all of $\mathbb{F}_2^{\ell \abs{\mathcal{M}}} = \mathbb{F}_2^{12}$.  For this example, this is clear by inspection, but we explain it in more detail to illustrate the general case. Consider re-arranging the rows and columns of $S$ as shown in Figure \ref{fig:example linear system remixed}. 
In this figure, we have grouped the rows and columns first by their monomials $\fm \in \cM$ or $\chi \in \cM_{\cL(r)}$, rather than first by their indices $i \in [\ell]$ or $r\in [n]$.
We observe that after this re-ordering, $S$ is a block diagonal matrix whose diagonal entries are $G(\Lambda)$ for some $\Lambda\subseteq [s] = [3]$ with $\abs{\Lambda} = 2$. Since each of the $G(\Lambda)$ is full-rank, it follows that $S$ is also full-rank.

\begin{figure}
    \centering   

  \begin{align*}
    S \mathbf{e} &:=\begin{array}{l}
1, a_1 \\
2, a_1 \\ \hline
1, a_2 \\
2, a_2 \\ \hline
1, a_3 \\
2, a_3 \\ \hline
1, b_1 \\
2, b_1 \\ \hline
1, b_2 \\
2, b_2 \\ \hline
1, b_3 \\ 
2, b_3
\end{array}
\left[
\begin{array}{cc|cc|cc|cc|cc|cc}
 0& 1& & & & &  & &  & & & \\
 1& 1& & & & & &  & & & & \\ \hline
 &  & 1 &1 & & & & & & & & \\
 & &  0 &1 & & & & & & &  & \\ \hline
 & &    &  & 1 & 0 & & & & & & \\
 & &    &  & 0 & 1 & & & & &  & \\ \hline
 &&&&&&0& 1& & & &  \\
 &&&&&&1& 1& & & &  \\ \hline
 &&&&&&&  & 1 &1 & &  \\
 &&&&&&& &  0 &1 & &  \\ \hline
 &&&&&&& &    &  & 1 & 0 \\
 &&&&&&& &    &  & 0 & 1 \\ 
\end{array}
\right]
\\&= \left[
\begin{array}{cccccc}
    G(\lbrace 2,3 \rbrace) & &&&&  \\
     & G(\lbrace 1,3 \rbrace) & &&& \\
     && G(\lbrace 1,2 \rbrace) &&&\\
    &&&G(\lbrace 2,3 \rbrace) & &  \\
     &&&& G(\lbrace 1,3 \rbrace) &  \\
     &&&&& G(\lbrace 1,2 \rbrace) 
\end{array}
\right]\mathbf{e}
\end{align*}
 \caption{Equivalent linear system obtained by permuting the rows and columns of $S$ from Figure~\ref{fig:example linear system}.}
    \label{fig:example linear system remixed}
\end{figure}

In summary, solving $S \cdot \mathbf{e} = \mathbf{g}$ for the coefficients $\mathbf{e}$ defines a function $\Eval$ that produces output shares $\mathbf{z} = (z_1, z_2, z_3)$ as in Equation \ref{eq: example output shares}, so that $G\mathbf{z} = (a,b)$ as desired. We have thus constructed a linear HSS $\pi = (\Share, \Eval, \Rec)$ with $\ell = 2$, $t =d=m= 1$, and $s = 3$ for $\text{POLY}_{1,1}(\mathbb{F}_2)^2$ from a high label-weight code with congruous parameters. 

\end{example}

\end{document}